\documentclass[11pt]{article}
\pdfoutput=1

\usepackage{fullpage}
\usepackage[ruled,vlined]{algorithm2e}
\usepackage{graphicx}   
\usepackage{subfig}  
\usepackage{hyperref}   
\usepackage{mdwlist}
\usepackage{xspace}
\usepackage[usenames,dvipsnames]{color}
\usepackage{amsmath, amsthm,amssymb, mathabx}

\usepackage{mathpazo} 
\usepackage[scaled]{helvet} 
\usepackage{courier} 
\normalfont
\usepackage[T1]{fontenc}

\newcommand{\cut}[1]{}

\newenvironment{packed_item}{
\begin{itemize}
   \setlength{\itemsep}{1pt}
   \setlength{\parskip}{0pt}
   \setlength{\parsep}{0pt}
}
{\end{itemize}}

\newenvironment{packed_enum}{
\begin{enumerate}
   \setlength{\itemsep}{1pt}
  \setlength{\parskip}{0pt}
   \setlength{\parsep}{0pt}
}
{\end{enumerate}}

\newenvironment{packed_grep}{
\begin{description}
   \setlength{\itemsep}{1pt}
   \setlength{\parskip}{0pt}
   \setlength{\parsep}{0pt}
}
{\end{description}}
              
\newcommand{\ie}{{\em i.e.}\xspace}
\newcommand{\eg}{{\em e.g.}\xspace}

\newcommand{\introparagraph}[1]{\textbf{#1.}}        

\newcommand{\set}[1]{\{#1\}}                    
\newcommand{\setof}[2]{\{{#1}\mid{#2}\}}        

\newcommand{\mM}[0]{\mathcal{M}}   
\newcommand{\mL}[0]{\mathcal{L}}   
\newcommand{\mR}[0]{\mathcal{R}}   
\newcommand{\mF}[0]{\mathcal{F}}   
\newcommand{\mW}[0]{\mathcal{W}}   

\newcommand{\bB}[0]{\mathbb{B}}  

\newcommand{\cqa}[1]{\textsc{Certainty}(#1)}  
\newcommand{\orset}[1]{\langle #1 \rangle}         

\newcommand{\mba}[0]{\mathbf{a}}   
\newcommand{\mbb}[0]{\mathbf{b}}   
\newcommand{\Eq}[0]{E^i /{\sim}}

\usepackage{aliascnt}  		

\newtheorem{theorem}{Theorem}[section]          	
\newaliascnt{lemma}{theorem}				
\newtheorem{lemma}[lemma]{Lemma}              	
\aliascntresetthe{lemma}  					
\newaliascnt{conjecture}{theorem}			
\newtheorem{conjecture}[conjecture]{Conjecture}    
\aliascntresetthe{conjecture}  				
\newaliascnt{remark}{theorem}				
              
\aliascntresetthe{remark}  					
\newaliascnt{corollary}{theorem}			
\newtheorem{corollary}[corollary]{Corollary}      
\aliascntresetthe{corollary}  				
\newaliascnt{definition}{theorem}			
\newtheorem{definition}[definition]{Definition}    
\aliascntresetthe{definition}  				
\newaliascnt{proposition}{theorem}			
\newtheorem{proposition}[proposition]{Proposition}  
\aliascntresetthe{proposition}  				
\newaliascnt{example}{theorem}			
\newtheorem{example}[example]{Example}  	
\aliascntresetthe{example}  				

\newcommand{\ux}[0]{\underline{x}}   
\newcommand{\uy}[0]{\underline{y}}   
\newcommand{\uz}[0]{\underline{z}}

\begin{document}

\title{A Dichotomy on the Complexity of Consistent Query Answering for Atoms with Simple Keys}

\author{Paraschos Koutris and Dan Suciu\\
\{pkoutris,suciu\}@cs.washington.edu\\
University of Washington}

\maketitle

\begin{abstract}
We study the problem of consistent query answering under primary key violations. In this setting, the relations in a database violate the key constraints and we are interested in maximal subsets of the database that satisfy the constraints, which we call repairs. For a boolean query $Q$, the problem CERTAINTY($Q$) asks whether every such repair satisfies the query or not; the problem is known to be always in coNP for conjunctive queries. However, there are queries for which it can be solved in polynomial time. It has been conjectured that there exists a dichotomy on the complexity of CERTAINTY($Q$) for conjunctive queries: it is either in PTIME or coNP-complete. In this paper, we prove that the conjecture is indeed true for the case of conjunctive queries without self-joins, where each atom has as a key either a single attribute (simple key) or all attributes of the atom. 
\end{abstract}

\maketitle

\section{Introduction}
\label{sec:introduction}

Uncertainty in databases arises in several applications and domains (\eg data integration, data exchange). An {\em uncertain} (or {\em inconsistent}) database is one that violates the integrity constraints of the database schema. In this work, we examine uncertainty under the framework of {\em consistent query answering}, established in~\cite{ABC99}.

In this framework, the presence of uncertainty generates many possible worlds, referred usually as {\em repairs}. For an inconsistent database $I$, a repair is a subset of $I$ that minimally differs from $I$ and also satisfies the integrity constraints. For a given query $Q$ on database $I$, the set of {\em certain answers} contains all the answers that occur in every $Q(r)$, where $r$ is a repair of $I$. The main research problem here is when the certain answers can be computed efficiently.  

In this paper, we will restrict the problem such that the integrity constraints are only {\em key constraints}, and moreover, the queries are {\em boolean conjunctive queries}. In this case, a repair $r$ of an inconsistent database $I$ selects from each relation a maximal number of tuples such that no two tuples are key-equal. We further say that a boolean conjunctive query $Q$ is {\em certain} if it evaluates to true for every such repair $r$. The decision problem \cqa{$Q$} is now defined as follows: given an inconsistent database $I$, does $Q(r)$ evaluate to true for every repair $r$ of $I$?

For this setting, it is known that \cqa{$Q$} is always in coNP~\cite{CM05}. However, depending on the key constraints and the structure of the query $Q$,  the complexity of the problem may vary. For example, for the query $Q_1 = R(\ux, y), S(\uy, z)$, \cqa{$Q_1$} is not only in P but, since one can show that \cqa{$Q_1$} can be expressed as a first-order query over $I$ \cite{FM07}, it is in $AC^0$. On the other hand, for $Q_2 = R(\ux, y), S(\uz, y)$, it has been proved in~\cite{FM07} that \cqa{$Q_2$} is coNP-complete. Finally, for $Q_3 = R(\ux, y), S(\uy, x)$, one can show~\cite{Wijsen10} that consistent query answering is in P, but the problem does not admit a first-order rewriting.

From the above examples, one can see that the complexity landscape is fairly intricate, even for the class of conjunctive queries. Although there has been progress in understanding the complexity for several classes of queries, the problem of deciding the complexity of \cqa{$Q$} remains open. In fact, a long-standing conjecture claims the following dichotomy.

\begin{conjecture}
\label{conj:dichotomy}
Given a boolean conjunctive query $Q$, \cqa{$Q$} is either in PTIME or is coNP-complete.
\end{conjecture}

The progress that has been made towards proving this conjecture has
been limited. In particular, Kolaitis and Pema~\cite{KP12} have proved a dichotomy into PTIME and coNP-complete for the case where $Q$ contains only two atoms and no self-joins (\ie every relation name appears once). Wijsen~\cite{Wijsen10b} has given a necessary and sufficient condition for first-order rewriting for acyclic conjunctive queries without self-joins, and in a recent paper~\cite{Wijsen13} further classifies several acyclic queries into PTIME and coNP-complete.

In this work, we significantly progress the status of the conjecture, by settling the dichotomy for a large class of queries: boolean conjunctive queries w/o self-joins, where each atom has as primary key either a single attribute or all the attributes. Observe that this class contains all queries where atoms have arity at most 2; in particular, it also contains all three of the queries $Q_1, Q_2, Q_3$ previously discussed.  Our results apply to a more general setting where one might have the external knowledge that some relations are {\em consistent} and others may  be {\em inconsistent}. In contrast to previous approaches, our paper introduces consistent relations since in non-acyclic queries, certain patterns in the structure of the query cause a relation to behave as a consistent relation when checking for certainty. In particular, consider a query $Q$ containing two atoms $R_1(\ux, y), R_2(\ux,y)$. If an instance contains the tuples $R_1(\underline{a},b_1), R_2(\underline{a},b_2)$ such that $b_1 \neq b_2$, we can remove the key-groups $R_1(\underline{a},-), R_2(\underline{a},-)$ without loss of generality in order to check for certainty\footnote{Indeed, if we want to find a repair $r$ that does not satisfy $Q$, we can always pick these two tuples to make sure that the value $a$ will never contribute to an answer.}. Thus, the conjunction of $R_1, R_2$ behaves as a single consistent relation $R(\ux,y)$. 

Our main result is

\begin{theorem}
\label{th:dichotomy}
For every boolean conjunctive query $Q$ w/o self-joins consisting only of binary relations where
exactly one attribute is the key, there exists a dichotomy of  \cqa{$Q$} into PTIME and coNP-complete.
\end{theorem}

From here we derive:

\begin{corollary} \label{cor:dichotomy}
For every boolean conjunctive query $Q$ with relations of arbitrary arity, where either exactly one attribute is a key, or the key consists of all attributes, there exists a dichotomy of \cqa{$Q$} into PTIME and coNP-complete.
\end{corollary}

Except for \autoref{sec:simplify:structure},  where we prove  \autoref{cor:dichotomy}, the rest of the  paper consists of the proof of \autoref{th:dichotomy}. The classification into PTIME and coNP-complete is  based on analyzing the structure of a specific graph representation of the query along with the key constraints. The {\em query graph}, which we denote $G[Q]$, is a directed graph with vertices the variables in $Q$, and a directed edge $(x,y)$ for every relation $R(\ux, y)$.

Given the graph $G[Q]$, we give a necessary and sufficient condition for \cqa{$Q$} to be computable in polynomial time. Consider two edges $e_R = (u_R, v_R), e_S = (u_S, v_S)$ in $G[Q]$ that correspond to two inconsistent relations $R$ and $S$ respectively. We say that $e_R, e_S$ are {\em source-equivalent} if $u_R, u_S$ belong to the same strongly connected component of $G[Q]$. We also say that $e_R, e_S$ are {\em coupled} if 
(a) there exists an undirected path $P_R$ from $v_R$ to $u_S$
 such that no node in $P_R$ is reachable from $u_R$ through a directed path in $G - \set{e_R}$ and 
(b) there exists an undirected path $P_S$ from $v_S$ to $u_R$ where no node in $P_S$ is reachable from $u_S$ through a directed path in $G - \set{e_S}$. Then:

\begin{theorem} 
\label{th:two:edges}
(1) \cqa{$Q$} is coNP-complete if $G[Q]$ contains a pair of inconsistent
edges that are coupled and not source-equivalent. Otherwise, \cqa{$Q$} is in PTIME. 
(2) The problem: given
a query $Q$ decide whether \cqa{$Q$} is coNP-complete or in PTIME is NLOGSPACE-complete.
\end{theorem}

The following example illustrates the main theorem.

\begin{example}
Consider the following two queries:
\begin{align*}
K_1 & = R(\ux,y), S(\uz,w), T^c(\uy,w)\\
K_2 & = R(\ux,y), S(\uz,w), T^c(\uy,w), U^c(\ux, z)
\end{align*}
Observe that the only difference between $K_1, K_2$ is the consistent relation $U^c$. Moreover, the edges $e_R, e_S$ are not source-equivalent in both cases. In $G[K_1]$, the edges $e_R, e_S$ are also coupled. Indeed, consider the path $P_R$ that consists of the edges $e_T, e_S$ and connects $y$ with $z$. The nodes $y,w,z$ of $P_R$ are not reachable from  $x$ in the graphs $G[K_1] - \set{e_R}$. Similarly, the path $P_S$ that consists of the edges $e_T, e_R$ connects $w$ with $x$ and is not reached by any directed path starting from $z$ in $G[K_1]-\set{e_S}$.  Thus, \cqa{$K_1$} is coNP-complete.

In contrast, the path $P_R$ is reachable from $x$ in $G[K_2]$: consider the path that consists of $e_U$. Since no other path connects $e_R, e_S$ in $G[K_2]$, the edges $e_R, e_S$ are not coupled. Thus, \cqa{$K_2$} is in PTIME.
\end{example}

Note that if two edges $e_R, e_S$ belong to two distinct weakly
connected components, then they are trivially not coupled, which
implies that $Q$ is coNP-complete iff one of its weakly connected
components is coNP complete.

 In order to show \autoref{th:two:edges}, we develop new techniques for efficient computation of \cqa{$Q$}, as well as techniques for proving hardness. We start by introducing in \autoref{sec:prelim} and \autoref{sec:general} the basic notions and definitions. In~\autoref{sec:cycles}, we present the case where $G[Q]$ is a strongly connected graph (\ie there is a directed path from any node to any other node) and show that \cqa{$Q$} is in PTIME. The algorithm for computing \cqa{$Q$} in this case is based on a novel use of {\em or-sets} to represent efficiently answers to repairs. The polynomial time algorithm for \cqa{$Q$} when $G[Q]$ satisfies the condition of \autoref{th:two:edges} is presented in \autoref{sec:general} and is based on a recursive decomposition of $G[Q]$. Finally, the hardness results are presented in \autoref{sub:coNP_case}, where we show that we can reduce the NP-hard problem \textsc{Monotone-3SAT} to any graph $G[Q]$ that does not satisfy the condition of \autoref{th:two:edges}.

\section{Preliminaries}
\label{sec:prelim}

A database schema is a finite set of relation names. Each relation $R$
has a set of attribute $attr(R) = \{A_1, \dots, A_k\}$, and a key,
which is a subset of $attr(R)$.  We write $R(\underline{x_1, \dots,
  x_m}, y_1, \dots, y_{\ell})$ to denote that the attributes on
positions $1, \dots, m$ are the primary key.  Each relation is of one
of two types: consistent, or inconsistent.  Sometimes we denote $R^c$ or
$R^i$ to indicate that the type of the relation is consistent or
inconsistent.

An instance $I$ consists of a finite relation $R^I$ for each relation
name $R$, such that, if $R$ is of consistent type, then $R^I$
satisfies its key constraint.  In other words, in an instance $I$ we
allow relations $R^i$ to violate the key constraints but always
require the relations $R^c$ to satisfy the key constraints.  Notice
that, if the key of $R$ consists of all attributes, then $R^I$ always
satisfies the key constraints, so we may assume w.l.o.g. that $R$ is
of consistent-type.

We denote a tuple by $R(\underline{a_1, \ldots, a_m},b_1, \ldots,
b_\ell)$.  We define a {\em key-group} to be all the tuples of a
relation with the same key, in notation $R(\underline{a_1, \ldots,
  a_m},-)$.

\begin{definition} [Repair] \label{def:repair}
  An instance $r$ is a {\em repair} for $I$ if (a) $r$ satisfies all
  key constraints and (b) $r$ is a maximal subset of $I$ that
  satisfies property (a).
\end{definition}

In this work, we study how to answer conjunctive queries on
inconsistent instances:

\begin{definition} [Consistent Query Answering]
  Given an instance $I$, and a conjunctive query $Q$, we say that a
  tuple $t$ is a {\em consistent answer} for $Q$ if for every repair
  $r \subseteq I$, $t \in Q(r)$. If $Q$ is a Boolean query, we say
  that $Q$ is certain for $I$, denoted $I \vDash Q$, if for every
  repair $r$, $Q(r)$ is true.

  If $Q$ is Boolean query, \cqa{$Q$} denote the following decision
  problem: given an instance $I$, check if $I \vDash Q$.
\end{definition} 

%

\subsection{Frugal Repairs}

\label{subsec:frugal-repairs}

Let $Q$ be a Boolean conjunctive query $Q$.  Denote $Q^f$ the {\em
  full query} associated to $Q$, where all variables become head
variables; therefore, for any repair $r$, $Q(r)$ is true iff $Q^f(r)
\neq \emptyset$.

\begin{definition} [Frugal Repair]
  A repair $r$ of $I$ is {\em frugal} for $Q$ if there exists no
  repair $r'$ of $I$ such that $Q^f(r') \subsetneq Q^f(r)$.
\end{definition}

\begin{example}
\label{ex:frugal}
Let $Q = R(\ux, y), S(\ux, y)$. 
In this case, the full query is $Q^f(x,y) = R(\ux, y), S(\ux, y)$. Also,
consider the instance 
\begin{align*} 
I  = \{ & R(\underline{a_1},b_1), R(\underline{a_1},b_2), R(\underline{a_2},b_3), 
S(\underline{a_1}, b_1), S(\underline{a_2}, b_3), \\
& R(\underline{a_3}, b_4), R(\underline{a_3}, b_5),
S(\underline{b_4}, a_3), S(\underline{b_5}, a_3) \}
\end{align*}
with the following repairs: 
\begin{align*}
r_1 & = \{R(\underline{a_1},b_1), R(\underline{a_2},b_3),
S(\underline{a_1}, b_1), S(\underline{a_2}, b_3),
 R(\underline{a_3}, b_4), 
S(\underline{b_4}, a_3), S(\underline{b_5}, a_3) \} \\
r_2 & = \{R(\underline{a_1},b_2), R(\underline{a_2},b_3),
S(\underline{a_1}, b_1), S(\underline{a_2}, b_3),
 R(\underline{a_3}, b_4), 
S(\underline{b_4}, a_3), S(\underline{b_5}, a_3) \}  \\
r_3 & = \{R(\underline{a_1},b_2), R(\underline{a_2},b_3),
S(\underline{a_1}, b_1), S(\underline{a_2}, b_3),
 R(\underline{a_3}, b_5), 
S(\underline{b_4}, a_3), S(\underline{b_5}, a_3) \}
\end{align*}
The repairs will produce the answer sets
$Q^f(r_1) = \{(a_1, b_1), (a_2, b_3), (a_3, b_4)\}$, 
$Q^f(r_2) = \{(a_2, b_3), (a_3, b_4)\}$ 
and $Q^f(r_2) = \{(a_2, b_3), (a_3, b_5)\}$  respectively. 
Since $Q^f(r_2) \subsetneq Q^f(r_1)$, the repair
$r_1$ is not frugal. On the other hand, both $r_2$ and $r_3$ are frugal.
\end{example}

\begin{proposition}
\label{prop:frugal_equivalence}
$I \vDash Q$ if and only if every frugal repair of $I$ for $Q$ satisfies $Q$. 
\end{proposition}

\begin{proof}
  One direction is straightforward: if some frugal repair does not
  satisfy $Q$, then $Q$ is not certain for $I$. For the other
  direction, assume that $Q$ is not certain for $I$.  Then there
  exists a repair $r$ s.t. $Q(r)$ is false, hence $Q^f(r)=\emptyset$:
  therefore $r$ is a frugal repair, proving the claim.   
\end{proof}

The proposition also implies that we lose no generality if we study
only frugal repairs in certain query answering.  To check $I \vDash Q$
it suffices to check whether $Q^f(r) \neq \emptyset$ for every frugal
repair.  In some cases, it is even possible to compute $Q^f(r)$ by
using a certain representation, as discussed next.

\subsection{Representability}

\label{subsec:orset}

In general, the number of frugal repairs is exponential in the size of
$I$.  We describe here a compact representation method for the set of
all answers $Q^f(r)$, where $r$ ranges over all frugal repairs.  We
use the notation of or-sets adapted from~\cite{LW93}. An {\em or-set}
is a set whose meaning is that one of its elements is selected
nondeterministically. Following \cite{LW93} we use angle brackets to
denote or-sets.  For example, $\orset{1,2,3}$ denotes the or-set that
is either 1 or 2 or 3; similarly $\orset{\set{1},\set{1,3}}$ means
either the set $\set{1}$ or $\set{1,3}$.

Let $\mF_Q(I) = \orset{r_1, r_2, \ldots}$ be the or-set of all frugal
repairs of $I$ for $Q$, and let
\begin{align*}
\mM_Q(I) = \orset{Q^f(r) \mid r \in \mathcal{F}_Q(I)}
\end{align*}
be the or-set of all answers of $Q^f$ on all frugal repairs.  Notice
that the type of $\mM_Q(I)$ is $\orset{\set{T}}$, where $T =
\bigtimes_{i=1}^k T_i$ is a product of atomic types.  For a simple
illustration, in \autoref{ex:frugal}, we have 
$\mM_Q(I) = \orset{\set{ (a_2, b_3), (a_3, b_4)}, \set{ (a_2, b_3), (a_3, b_5)}}$, 
because $r_2, r_3$ are the only frugal repairs.

Give a type $T$, define the following function $\alpha :
\set{\orset{T}} \rightarrow \orset{\set{T}}$ \cite{LW93}:
$\alpha(\set{A_1, \ldots, A_m}) = \orset{\set{x_1, \ldots, x_m} | x_1
  \in A_1, \ldots, x_m \in A_m}$.  For example,
$\alpha(\set{\orset{1,2}, \orset{3,4}}) = $ $\orset{\set{1,3},
  \set{1,4}, \set{2,3}, \set{2,4}}$ and
$\alpha(\set{\orset{1},\orset{1,2,3}}) =
\orset{\set{1},\set{1,2},\set{1,3}}$.

\begin{definition}
\label{def:representable}
Let $T = \bigtimes_{i=1}^k T_i$.  An or-set-of-sets $S$ (of type
$\orset{\set{T}}$) is {\em representable} if there exists a
set-of-or-sets $S_0$ (of type $\set{\orset{T}}$) such that (a)
$\alpha(S_0) = S$ and (b) for any distinct or-sets $A, B \in S_0$, the
tuples in $A$ and $B$ use distinct constants in all coordinates:
$\Pi_i(A) \cap \Pi_i(B) = \emptyset$, $\forall i = 1, k$.
\end{definition}

As an example, consider the or-sets 
\begin{align*}
S& = \orset{ \set{(a_1,b_1), (a_2,b_3)},\set{(a_1,b_2), (a_2, b_3)}} \\
S' & = \orset{ \set{(a_1,b_1), (a_2,b_3)},\set{(a_1,b_2), (a_2, b_2)}}
\end{align*}

$S$ is representable, since we can find a compression $S_0 =
\set{\orset{(a_1, b_1), (a_1,b_2)}, \orset{(a_2, b_3)}}$. Notice that
$a_1, b_1, b_2$ appear only in the first or-set of $S_0$, whereas
$a_2,b_3$ only in the second. On the other hand, it is easy to see
that $S'$ is not representable.
We prove:

\begin{proposition}
\label{prop:compression:size}
Let $S$ be an or-set of sets of type $\orset{\set{\bigtimes_{i=1}^k
    T_i}}$, and suppose that its active domain has size $n$.  If $S$
is representable $S = \alpha(S_0)$, then its compression $S_0$ has
size $O(n^k)$.
\end{proposition}

\begin{proof}
  If $S_0 = \set{A_1, A_2, \ldots}$, then every $k$-tuple consisting
  of constants from the active domain occurs in at most one or-set,
  thus the total size of $S_0$ is $O(n^k)$.
\end{proof}

If $\mM_Q(I)$ is representable, then we denote $A_Q(I)$ its
compression; its size is at most polynomially large in $I$.  In
general, $\mM_Q(I)$ may not be representable.

By the definition of frugality, if $s_1, s_2 \in \mM_Q(I)$ then
neither $s_1 \subsetneq s_2$ nor $s_2 \subsetneq s_1$ holds.  This
implies that, for any instance $I$, there are two cases.  Either (1)
$I \not \vDash Q$; in that case $\mM_Q(I) = \orset{\set{}}$ is
trivially representable as $A_Q(I) = \set{}$; or, (2) $I \vDash Q$,
and in that case $\mM_Q(I) = \orset{A_1, A_2, \ldots}$, where $A_i\neq
\set{}$ for all $i$, may be exponentially large and not necessarily
representable.  For a simple illustration, in \autoref{ex:frugal},
$\mM_Q(I) $ is representable, and its
compression is $A_Q(I) = \set{\orset{(a_2, b_3)}, \orset{(a_3, b_4), (a_3, b_5)}}$.

If $A_Q(I)$ exists for every instance $I$ and can be computed in
polynomial time in the size of $I$, then \cqa{$Q$} is PTIME: to check
$I \vDash Q$, simply compute $A_Q(I)$ and check $\neq \set{}$.  The
converse is not true, however: for example, consider the query $H =
R(\underline{x},y), S(\underline{y},z)$, for which \cqa{$H$} is in
PTIME.  However, for the instance $I' = \set{R(\underline{a},b),
  S(\underline{b}, c_1), S(\underline{b},c_2)}$, $\mM_H(I') =
\orset{\set{(a,b,c_1)}, \set{(a,b,c_2)}}$ is not representable.

\subsection{Purified Instances} 
Let $Q$ be a any boolean conjunctive query.  An instance $I$ is called
{\em globally consistent}~\cite[pp.128]{DBLP:books/aw/AbiteboulHV95},
or {\em purified}~\cite{Wijsen13}, if for every relation $R$,
$\Pi_{attr(R)}(Q^f(I)) = R^I$, where $\Pi_{attr(R)}$ denotes the
projection on the attributes of relation $R$.  In other words, no
tuple in $I$ is ``dangling''.  

In the rest of the paper we will assume that the instance $I$ is
purified.  This is without loss of generality, because if $I$ is an
arbitrary instance, then we can define a new instance $I^p \subseteq I$
such that $\mM_Q(I) = \mM_Q(I^p)$, and thus $I \vDash Q$ if and only if
$I^p \vDash Q$.
\begin{lemma}
  Given a query $Q$ and an instance $I$, there exists a purified instance $I^p \subseteq I$
  such that $\mM_Q(I) = \mM_Q(I^p)$.
\end{lemma}

\begin{proof}
If $I$ is not purified, there exists a tuple $t$ in the key-group $R(\underline{a_1, \dots, a_m},-)$ 
such that $t \notin \Pi_{attr(R)}(Q^f(I))$. Then, for any frugal repair $r$ of $Q$, no tuple from 
$R(\underline{a_1, \dots, a_m},-)$ will contribute to some tuple in $Q^f(r)$; otherwise, for the repair
$r' = r \setminus R(\underline{a_1, \dots, a_m},-) \cup \set{t}$ we would have $Q^f(r') \subset Q^f(r)$.
Thus, for $I' = I \setminus R(\underline{a_1, \dots, a_m},-)$, we have $\mM_Q(I) = \mM_Q(I')$. We repeat this process until we get a purified instance $I^p$.
\end{proof}

\subsection{The Query Graph}

In the rest of the paper we will restrict the discussion to the
setting of \autoref{th:dichotomy}, and consider only Boolean queries
w/o self-joins consisting only of binary relations where exactly one
attribute is the key; in \autoref{sec:simplify:structure} we prove
\autoref{cor:dichotomy}, thus extending the dichotomy to more general
queries.


Given a query $Q$, the {\em query graph} $G[Q]$ is a directed graph
where the vertex set $V(G)$ consists of set of variables in $Q$, and
edge set $E(G)$ contains for atom $R(\underline{u}, v)$ in $Q$ an edge
$e_R = (u,v)$ in $G[Q]$.  Since $Q$ has no self-joins each relation
$R$ defines a unique edge $e_R$, and we denote $u_R$ and $v_R$ its
starting and ending node respectively.  We say that the edge is
consistent (inconsistent) if the type of $R$ is consistent
(inconsistent), and denote $E^i(G)$ ($E^c(G)$) the set of all
consistent (inconsistent) edges. Thus $E(G) = E^i(G) \cup E^c(G)$.

A {\em directed path} $P$ is an alternating sequence of vertices and edges
$v_0, e_1, v_1, \dots, e_{\ell}, v_{\ell}$ where $e_i = (v_{i-1},v_i)$
for $i=1, \dots, \ell$ and $\ell \geq 0$. We write $P : x \rightarrow y$ for a
directed path $P$ where $v_0 = x$ to $v_{\ell} = y$, and every edge $e_i$ is
consistent; we write $P : x \leadsto y$ for any directed path $P$ where
$v_0 = x$ and $v_{\ell} = y$ that has any type of edges. An {\em undirected
path} $P$ is an alternating sequence of vertices and edges 
$v_0, e_1, v_1, \dots, e_{\ell}, v_{\ell}$ where either $e_i = (v_{i-1},v_i)$
or $e_i = (v_i, v_{i-1})$ for $i=1, \dots, \ell$ and $\ell \geq 0$; we write
$P : x \leftrightarrow y$ for an undirected path where $v_0 = x$ and $v_{\ell} = y$ 
(that may also have any types of edges).
A path $P$ may contain a single vertex and no edges (when $\ell = 0$), in which
case we can write $P: x \rightarrow x$. If $N \subseteq V(G)$, then $P \cap N$
denotes the set of vertices in $P$ that occur in $N$.  The notation $x
\rightarrow y$ (or $x \leadsto y$, or $x \leftrightarrow y$) means
``there exists a path $P : x \rightarrow y$'' (or $P: x \leadsto y$,
or $P: x \leftrightarrow y$).


Finally, since $Q$ uniquely defines $G[Q]$ and vice versa, we will
often use $G$ to denote the the query $Q$ (for example, we may
say $G(r)$ instead of the boolean value $Q(r)$, for some repair $r$).

\begin{example}
\label{ex:q_graph}
Consider the following query:
\begin{align*}
H = & R_1(\ux,y), R_2^c(\uy,z), R_3(\underline{z},x), 
 V_1^c(\underline{u},y), V_2^c(\underline{x},v), \\
 & V_3^c(\underline{z},v), 
 S(\underline{u},v), T(\underline{v}, w), U^c(\underline{u},w)
\end{align*}
The graph $G[H]$ is depicted in \autoref{fig:q_graph}. The curly edges denote inconsistent edges $E^i = \set{R_1, R_3,S,T}$, whereas the straight edges denote consistent ones. 
We also have $u \leadsto x$ (but not $u \rightarrow x$, since the only path from $u$ to $x$ contains inconsistent edges). Moreover, $y \rightarrow v$, since there is a directed path that goes from $y$ to $v$ through $R_2, V_3$. Finally, notice that, although $v \not \leadsto y$, $v \leftrightarrow y$. 
\end{example}

\begin{figure}[tb]
\centering{ \resizebox{6cm}{!} {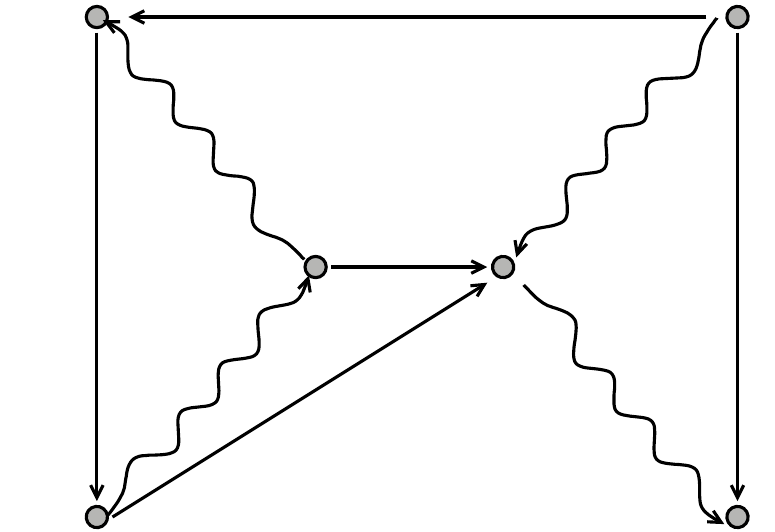}}
\caption{{\footnotesize The query graph $G[H]$. The curly edges denote inconsistent relations, whereas the straight edges consistent relations.}}
\label{fig:q_graph}
\end{figure}

\subsection{The Instance Graph}

Let $Q$ be a Boolean conjunctive query without self-joins over binary
relations with single-attribute keys.  Let $I$ be an instance for $Q$.
We will assume w.l.o.g. that any two attributes that are not joined by
$Q$ have disjoint domains: otherwise, we simply rename the constants
in one attribute.  For example, if $Q = R(x,y),S(y,z),T(z,x)$ then we
will assume that $\Pi_1(R^I) \cap \Pi_1(S^I)=\emptyset$, etc.

The {\em instance graph} is the following directed graph $F_Q(I)$.
The nodes consists of all the constants occurring in $I$, and there is
an edge $(a,b)$ for every tuple $R^I(\underline{a}, b)$ in $I$.  The
size of the instance graph $F_Q(I)$ is the same as the size of the
instance $I$.

\section{The Dichotomy Theorem}
\label{sec:general}

We present here formally our dichotomy theorem, and start by
introducing some definitions and notations.  Let $u \in V(G)$ and $e_R
\in E(G)$. Then,
\begin{align*}
u^{\oplus} &= \setof{v \in V(G)}{u \rightarrow v \text{ in } G} \\
u^{+} &= \setof{v \in V(G)}{u \leadsto v \text{ in } G} \\
u^{+,R} & = \setof{v \in V(G)}{u \leadsto v \text{ in } G - \{e_R\}}
\end{align*}

\begin{example}
  Consider the graph $G[H]$ from \autoref{fig:q_graph}, which will be
  our running example.  Then:
  \begin{align*}
    x^{\oplus} = & \set{x,v} & x^{+,R_1} = & \set{x,v,w} & x^+ =& \set{x,v,w,y,z}
  \end{align*}
\end{example}

\begin{proposition}
\label{prop:plus_contain}
If $R \in E^i$, $  u_R^{\oplus} \subseteq u_R^{+,R} \subseteq u_R^{+}$. 
\end{proposition}

\begin{proof}
Let $v \in u_R^{\oplus}$. Then, there exists a path $P: u_R \rightarrow v$ in $G$.
Since $P$ is consistent, it cannot contain the inconsistent edge $e_R$, and thus $P$ exists in
$G - \{e_R\}$ as well. Consequently, $v \in u_R^{+,R}$. The other inclusion is straightforward.
\end{proof}

Define the binary relation $R \lesssim S$ if $u_S \in u_R^+$.  The
relation $\lesssim$ is a {\em preorder} the set of edges, since it is
reflexive and transitive.
If $R \lesssim S$ and $S \lesssim R$ then we say that $R,S$ are {\em
  source-equivalent} and denote $R \sim S$.  Notice that $R \sim S$
iff their source nodes $u_R, u_S$ belong in the same strongly
connected component (SCC) of $G$; in particular, if $R, S$ have the
same source node, $u_R = u_S$, then $R \sim S$.

For an edge $R \in E^i$, we define the following sets of {\em coupled}
edges:
\begin{align*}
  coupled^{\oplus}(R) & = [R] \cup \setof{S \in E^i}{\exists P: v_R  \leftrightarrow u_S, P  \cap u_R^{\oplus}  = \emptyset}\\
  coupled^{+}(R) & = [R] \cup \setof{S \in E^i}{\exists P: v_R  \leftrightarrow u_S, P  \cap u_R^{+,R}  = \emptyset}
\end{align*}

By definition, every edge $S$ that is source-equivalent to $R$ is
coupled with $R$.  In addition, 
$coupled^{\oplus}(R)$ ($coupled^{+}(R)$), includes all inconsistent
edges $S$ whose source node $u_S$ is in the same weakly connected
component as $v_R$, in the graph $G-u_R^{\oplus}$ ($G- u_R^{+,R}$
respectively). The notion of $coupled^{\oplus}$ is not necessary to express
the dichotomy theorem, but it will be heavily used in the polynomial time algorithm
of \autoref{sub:ptime_case}.

\begin{example} 
  Let us compute the coupled edges in our running example, where $E^i
  = \set{R_1, R_3, S, T}$.  We start by computing the node-closures of
  all the four source nodes:
  \begin{align*}
x^{\oplus} = & \set{x,v} & x^{+,R_1} = & \set{x,v,w} \\
z^{\oplus} = & \set{z,v} & z^{+,R_3} = & \set{z,v,w} \\
u^{\oplus} = & \set{u, y, w} & u^{+,S} = & \set{u, y,  w, x, v, w} \\
v^{\oplus} = & \set{v} & v^{+,T} = & \set{v}
  \end{align*}


  Next, we compute $coupled^+(e)$ for every inconsistent edge $e$. For
  example, the set $coupled^+(R_1)$ includes $R_1$ and $R_3$, because
  $R_1 \sim R_3$.  In addition, after we remove $x^{+,R_1} =
  \set{x,v,w}$ from the graph, the destination node $y$ of $R_1$ is
  still weakly connected to the source node $u$ of $S$, thus
  $coupled^+(R_1)$ contains $S$; but $y$ is no longer connected to the
  source node $v$ of $T$, therefore $coupled^+(R_1)$ does not contain
  $T$.  By similar reasoning:
  \begin{align*}
    coupled^\oplus(R_1) = & \set{R_1, R_3, S} && coupled^+(R_1) = \set{R_1, R_3, S} \\
    coupled^\oplus(R_3) = & \set{R_1, R_3, S, T} && coupled^+(R_3) = \set{R_3} \\
    coupled^\oplus(S) = & \set{S} && coupled^+(S) = \set{S} \\
    coupled^\oplus(T) = & \set{R_1, R_3, S, T} && coupled^+(T) =  \set{R_1, R_3, S, T}
  \end{align*}
\end{example}

\autoref{prop:plus_contain} implies:

\begin{corollary}
If $R \in E^i$, $coupled^{\oplus}(R) \supseteq coupled^+(R)$.
\end{corollary}

\begin{definition}
  Two edges $R,S \in E^i$ are {\em coupled} if $R \in coupled^+(S)$
  and $S \in coupled^+(R)$.

  The graph $G$ is called {\em unsplittable} if there exists two
  coupled edges $R, S$ s.t. $R \not\sim S$.  Otherwise, the graph is
  called {\em splittable}. 

\end{definition}

The graph $G[H]$ from our running example is splittable, because the
only pair of coupled edges are $R_1, R_3$, which are also
source-equivalent $R_1 \sim R_3$.  Indeed, any other pair is not
coupled: $R_1, S$ are not coupled because $R_1 \not\in coupled^+(S)$;
$R_1, T$ are not coupled because $T \not\in coupled^+(R_1)$; etc.

We can now state our dichotomy theorem, which we will prove in the rest of the paper.  

\begin{theorem} [Dichotomy Theorem]
\label{thm:dichotomy}
(1) If $G[Q]$ is splittable, then \cqa{$Q$} is in PTIME.  (2) If
$G[Q]$ is unsplittable, then \cqa{$Q$} is coNP-complete.
\end{theorem}

We end this section with a few observations. First, if $Q$ consists of
several weak connected components $Q_1, Q_2, \ldots$, in other words,
$Q_i, Q_j$ do not share any variables for all $i\neq j$, then $Q$ is
unsplittable iff some $Q_i$ is unsplittable: this follows from the
fact that $coupled^+(R)$ is included in the weakly connected component
$Q_i$ that contains $R$.  In this case, \autoref{thm:dichotomy}
implies that \cqa{$Q$} is coNP-complete iff  \cqa{$Q_i$} is
coNP-complete for some $i$.

Second, if $Q$ is strongly connected, then it is, by definition,
splittable: in this case \autoref{thm:dichotomy} says that \cqa{$Q$}
is in PTIME.  In fact, the first step of our proof is to show that
every strongly connected query is in PTIME.

Finally, we note that the property of being splittable or unsplittable
may change arbitrarily, as we add more edges to the graph.  For
example, consider these three queries: $Q_1 = R(\underline{x},y)$,
$Q_2 = R(\underline{x},y),S(\underline{z}, y)$, $Q_3 =
R(\underline{x},y),S(\underline{z}, y),T(\underline{z}, y)$, where all
three relations $R, S, T$ are inconsistent.  Then $Q_1, Q_3$ are
splittable, while $Q_2$ is unsplittable, and therefore, their
complexities are PTIME, coNP-hard, PTIME.  Indeed, in $Q_2$ we have
$coupled^+(R)= coupled^+(S) = \set{R, S}$, therefore $R, S$ are
coupled and in-equivalent $R \not\sim S$, thus, $Q_2$ is unsplittable.
On the other hand, in $Q_3$ we have\footnote{The difference between
  $Q_2$ and $Q_3$ is that in $Q_2$ we have $z^{+,S}=\set{z}$, while in
  $Q_3$ we have $z^{+,S}=\set{x,y,z}$.} $coupled^+(S) = \set{S,T}$,
$coupled^+(T)= \set{S,T}$, and therefore $R,S$ are no longer coupled,
nor are $R,T$: $Q_3$ is splittable.

\section{Strongly Connected Graphs}
\label{sec:cycles}

If $G[Q]$ is a strongly connected graph (SCG), then it is, by
definition, splittable.  Our first step is to prove Part (1) of
\autoref{thm:dichotomy} in the special case when $G[Q]$ is a strongly
connected, by showing that \cqa{$Q$} is in PTIME.  We actually show an
even stronger statement.

\begin{theorem} \label{th:scc}
  If $G[Q]$ is strongly connected, $\mM_{Q}(I)$ is representable and
  its compression $A_Q(I) $can be computed in polynomial time in the
  size of $I$.
\end{theorem}

As we discussed in~\autoref{sec:prelim}, \cqa{$Q$} is false if and only if
$A_Q(I) = \set{}$; hence, as a corollary of the theorem we obtain:

\begin{corollary}
If $G[Q]$ is strongly connected, \cqa{$Q$} is in PTIME.
\end{corollary}

We start in~\autoref{subsec:cycles} by proving \autoref{th:scc} in the
special case when $G[Q]$ is a directed cycle; we prove the
general case in~\autoref{subsec:scc}.

\subsection{A PTIME Algorithm for Cycles}
\label{subsec:cycles}

For any $k\geq 2$, the cycle query $C_k$ is defined as:
\begin{align*}
  C_k = R_{1}(\ux_{1},x_{2}), R_{2}(\ux_{2},x_{3}), \dots,  R_{k}(\ux_{k}, x_{1})
\end{align*}
Wijsen~\cite{Wijsen13} describes a PTIME algorithm for computing
\cqa{$C_2$}.  We describe here a PTIME algorithm for computing
$A_{C_k}(I)$ (and thus for computing \cqa{$C_k$} for arbitrary $k
\geq 2$ as well), called \textsc{FrugalC}. 

\begin{figure}[tb]
  \centering
  {{ \renewcommand\tabcolsep{5pt}
        \begin{tabular}{|c|c|c|c|c|}
            \multicolumn{1}{l}{$R(\ux,y)$} & \multicolumn{1}{l}{ } & 
            \multicolumn{1}{l}{$S(\uy, z)$} & \multicolumn{1}{l}{ } & 
            \multicolumn{1}{l}{$T(\uz,x)$}  \\
            
            \cline{1-1}\cline{3-3} \cline{5-5}                                         
             $(a_1, b_1)$ &  & $(b_1, c_1)$ &  & $(c_1, a_1)$   \\
             $(a_1, b_2)$ &  & $(b_2, c_1)$ &  &     \\
             $(a_2, b_2)$ &  & $(b_2, c_2)$ &  & $(c_2, a_2)$    \\
                  &  &  &  & \\
             $(a_3, b_3)$ &  & $(b_3, c_3)$ &  & $(c_3, a_3)$  \\
             $(a_3, b_4)$ &  & $(b_4, c_4)$  &  & $(c_4, a_3)$ \\
             $(a_4, b_4)$ &  & $(b_4, c_3)$ &  & $(c_3, a_4)$   \\
               
            \cline{1-1}\cline{3-3} \cline{5-5} 
          \end{tabular}}}
\caption{{\footnotesize An inconsistent purified instance $I$ for $C_3$.}}
\label{fig:cycle_instance}
\end{figure}

\begin{lemma} 
\label{lemma:purified_scc}
Let $I$ be a purified instance relative to $C_k$. Then, the instance
graph $F_{C_k}(I)$ is a collection of disjoint SCCs such that every
edge has both endpoints in the same SCC.
\end{lemma} 

\begin{proof}
  Let $(u,v)$ be a directed edge in the graph. Since $I$ is purified,
  $(u,v)$ must belong in a cycle and thus there exists a directed path
  $v \rightarrow u$, implying that $u,v$ are in the same SCC.
\end{proof}    

\introparagraph{Algorithm} Fix $k \geq 2$.  The algorithm
\textsc{FrugalC} takes as input a purified instance $I$ and returns
the compression $A_{C_k}(I)$ of $\mM_{C_k}(I)$, in four steps:

\begin{packed_enum}
\item \label{item:alg:1} Compute the SCCs of $F_{C_k}(I)$: $F_{C_k}(I) = F_1 \cup \ldots \cup F_m$,
where each $F_i$ is an SCC, and there are no edges between $F_i, F_j$
for $i\neq j$.
\item \label{item:alg:2} Compute $S = \setof{i}{F_i \mbox{ has no cycle of length } >k}$.
\item \label{item:alg:3} For each $i \in S$, define the or-set: $A_i =  \orset{(a_1, \dots, a_k) \mid a_1, \dots, a_k \mbox{ cycle in $F_i$}}$.
\item \label{item:alg:4} Return: $\setof{A_i}{i \in S}$.
\end{packed_enum}

Step \ref{item:alg:1} is clearly computable in PTIME.  In Step
\ref{item:alg:2}, we remove all SCC's $F_i$ that contain a cycle of
length $> k$: to check that, enumerate over all simple paths of length
$k+1$ in $F_i$ (there are at most $O(n^{k+1})$), and for each path
$u_0, u_1, u_2, \ldots, u_k$ check whether there exists a path from
$u_k$ to $u_0$ in $F_i - \set{u_1, \ldots, u_{k-1}}$.  After Step
\ref{item:alg:2}, if $i \in S$, then every cycle in $F_i$ has length
$k$, and every edge is on a $k$-cycle (because $I$ is purified).  Step
\ref{item:alg:3} constructs an or-set $A_i$ consisting of all
$k$-cycles of $F_i$ (there are at most $O(n^k)$).  The last step
returns the set of all or-sets $A_i$: this is a correct representation
(\autoref{def:representable}) because no two or-sets $A_i, A_j$ have
any common constants (since they represent cycles from different
SCC's).  We will prove in the rest of the section that $A_{C_k}(I) =
\setof{A_i}{i \in S}$, and therefore the algorithm correctly computes
$A_{C_k}(I)$.  Note that $I \vDash C_k$ iff $A_{C_k}(I)= \set{}$ iff
$S=\emptyset$.

%
%


\begin{figure}[tb]
  \centering
\resizebox{7cm}{!}{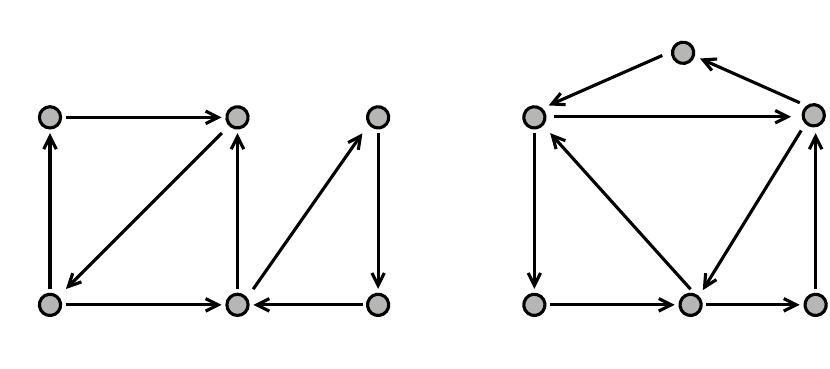}
\caption{{\footnotesize{The graph $F_{C_3}(I)$ for the instance in \autoref{fig:cycle_instance} has two SCC's, $F_1$ and $F_2$.}}}
\label{fig:cycle_i0}
\end{figure}

\begin{example} We illustrate the algorithm on $C_3 = R(\ux,y), S(\uy,
  z), T(\uz,x)$.  Consider the relations $R,S,T$ of the instance $I$
  in \autoref{fig:cycle_instance} and its graph $F_{C_3}(I) = F_1 \cup
  F_2$ shown in \autoref{fig:cycle_i0}.  The SCC $F_1$ contains only
  cycles of length 3: $(a_1,b_1,c_1)$, $(a_1,b_2,c_1)$ and $(a_2,
  b_2,c_2)$, whereas $F_2$ contains a cycle\footnote{Notice that {\em
      every} edge in $F_2$ is on some cycle of length 3 (since $I$ is
    purified), yet $F_2$ also contains a cycle of length
    6.} of length 6: $(a_3, b_3, c_3,
  a_4, b_4, c_4)$. Therefore the algorithm returns a set consisting of
  a single or-set:
$$A_{C_3}(I) = \{\orset{(a_1, b_1, c_1), (a_1, b_2, c_1), (a_2, b_2, c_2)}\}$$
\end{example}

It remains to show that the algorithm is correct, and this follows
from two lemmas.  Recall from \autoref{subsec:orset} that
$\mF_{C_k}(I)$ denotes the or-set of frugal repairs of $I$ for $C_k$.
Assuming $I$ is a purified instance, let $I = I_1 \cup I_2 \cup \ldots
\cup I_m$, where each $I_i$ corresponds to some SCC of $F_{C_k}(I)$.

\begin{lemma} 
\label{lemma:many:cycles}
  $\mF_{C_k}(I) = \orset{r_1 \cup \ldots \cup r_m | r_1 \in \mF_{C_k}(I_1), \ldots, r_m \in \mF_{C_k}(I_m)}$
\end{lemma}

In other words, the frugal repairs of $I$ are obtained by choosing,
independently, a frugal repair $r_i$ for each SCC $I_i$, then taking
their union.


\begin{lemma} \label{lemma:single:cycle}
Let $I$ be a purified instance relative to $C_k$, such that $F_{C_k}(I)$ is strongly connected. Then:
\begin{equation*}
\mM_{C_k}(I) = 
\begin{cases} \orset{\set{}} \quad \text{if $I$ has a cycle of length $>k$,} \\
\orset{ \set{(a_1, \dots, a_k)} \mid a_1, \dots, a_k \text{ cycle in $F_{C_k}(I)$}}  \quad \text{ otherwise}
\end{cases}
\end{equation*}
\end{lemma}
 
The lemma says two things.  On one hand, if $I$ has a cycle of length
$>k$, then $I \not\vDash C_k$.  Consider the case when all cycles in
$I$ have length $k$.  In general, if $r$ is a minimal repair, then the
full query $C_k^f(r)$ may return any nonempty set of $k$-cycles.  The
lemma states that if $r$ is a frugal repair, then $C_k^f(r)$ returns
{\em exactly one} $k$-cycle, and, moreover, that every $k$-cycle is
returned on some frugal repair $r$.

\begin{proof}
  To simplify the notation, we denote $F_{C_k}(I)$ by $F(I)$.  The
  lemma follows from the following claim: for any cycle $C$ in $F(I)$,
  there exists a repair $r_C \subseteq I$ such that $F(r_C)$ contains
  only $C$ as a cycle.  Indeed, if $I$ has some cycle $C$ of length
  $>k$, then the query $C_k$ is false on $r_C$, proving that
  $I\not\vDash C_k$; otherwise, for every cycle $C$ of length $k$,
  $C_k^f(r_C)$ returns only that cycle, and therefore $\mM_Q(I)$ is an
  or-set of singleton sets of the form $\set{C}$, for every $k$-cycle
  $C$.  Thus, it remains to prove the claim.

  Any subset $r \subseteq I$ represents a subset of edges of $F(I)$.
  Denote $V(r)$ the set of constants in $r$, and denote $K(r)$ the set
  of constants that occur in key positions in $r$, i.e. $K(r) =
  \setof{a}{\exists R(\underline{a},b) \in r}$; obviously, $K(r)
  \subseteq V(r)$.  Note that $K(I)=V(I)$ because $I$ is purified.

  To prove the claim, fix a cycle $C$, and define a strictly
  increasing sequence of instances $r_0 \subset r_1 \subset \ldots
  \subset r_\ell \subseteq I$ such that for every $i$: (1) $r_i$ is
  consistent (i.e. it satisfies all key constraints), (2) $V(r_i) =
  K(r_i)$, (3) $r_i$ contains only $C$ as a cycle.  In addition,
  $r_\ell$ is a repair (\autoref{def:repair}).  Then, the claim
  follows by setting $r_C = r_\ell$.

  We start the sequence by setting $r_0 = C$.  Clearly $r_0$ satisfies
  all key constraints 
   and $K(r_0)=V(r_0)$. 
%
%
  Now, consider some $r_i$ for $i\geq 0$.  If $K(r_i) = V(I)$ then
  $r_i$ is a repair (\autoref{def:repair}) and we stop, setting
  $\ell=i$.  Otherwise, let $V'= V(I)-K(r_i)$.  Since $F(I)$ is
  strongly connected, there exists an edge from $V'$ to $K(r_i)$, in
  other words, there exists a tuple $R(\underline{a}, b)$ such that
  $a\in V'$ and $b \in K(r_i)$.  Define $r_{i+1} = r_i \cup
  \set{R(\underline{a}, b)}$.  We check the three properties. (1)
  $r_{i+1}$ is consistent, because $a$ did not occur as a key in
  $r_i$.  (2) $V(r_{i+1}) = V(r_i) \cup \set{a}$ and $K(r_{i+1})=
  K(r_i) \cup \set{a}$; by induction we conclude $V(r_{i+1})=
  K(r_{i+1})$.  (3) Let $C'\neq C$ be a cycle in $r_{i+1}$, then $C'$
  must include the new edge $(a,b)$ (since $r_i$ has only $C$ as
  cycle).  Then the preceding edge $(c,a)$ must be in $r_i$, which is
  a contradiction because $a \not\in K(r_i) = V(r_i)$.
\end{proof} 

We now apply the two lemmas to prove the correctness of the algorithm.
\autoref{lemma:single:cycle} implies that, if $I$ is strongly
connected and has no cycle of length $>k$, then $\mM_{C_k}(I)$ is
represented by $A_{C_k}(I) = \set{\orset{(a_1, \ldots, a_k)|a_1,
    \ldots, a_k \mbox{ a cycle in } F_{C_k}(I)}}$; and if $I$ has a
cycle of length $>k$ then $A_{C_k}(I)= \set{}$.
\autoref{lemma:many:cycles} implies that, if $I$ has $m$ SCC's $I=I_1
\cup \ldots \cup I_m$, then $A_{C_k}(I) = A_{C_k}(I_1) \cup \ldots
A_{C_k}(I_m)$.  This completes the correctness proof of the the
algorithm.

We conclude this section with an observation
on FO-expressibility.  Recall that~\cite{Wijsen10} proves that the \cqa{$C_2$} is not first-order (FO)-expressible. The following proposition completes the complexity landscape for cycle queries.

\begin{proposition}
\label{prop:fo_expressible}
For a cycle query $C_k$ ($k>1$), \cqa{$C_k$} is FO-expressible if and only if $C_k$ contains at most one inconsistent edge.
\end{proposition}


\subsection{A PTIME Algorithm for SCGs}
\label{subsec:scc}

We now present the general algorithm that computes the compression
$A_Q(I)$ for any strongly connected query $Q$.  The algorithm uses the
following  decomposition of the query graph $G[Q]$.

Let $G = G[Q]$ be a query graph and $G_0 \subseteq G$ be subgraph.  A
{\em chordal path} for $G_0$ is a simple, non-empty\footnote{Recall
  that, when $u=v$, then a simple, non-empty path from $u$ to $u$ is a
  cycle.} path $P : u \leadsto v$ s.t. $G_0 \cap P = \set{u,v}$.  
If $P$ consists of a single edge then we call it a {\em chord}.  With
some abuse, we apply the same terminology to queries: if the query $Q$
can be written as $Q_0, P$, where $Q_0$ and $P$ are sets of atoms s.t
$P$ is a simple path\footnote{Meaning that $P =
  R_1(\underline{u},x_1), R_2(\underline{x_1},x_2), \ldots,
  R_m(\underline{x_{m-1}},v)$, all variables $u, x_1, \ldots, x_{m-1}$
  are distinct, and all variables $x_1, \ldots, x_{m-1}, v$ are
  distinct.}  from $u$ to $v$, then we say that $P$ is a chordal path
for $Q_0$ if they share only the variables $u,v$.

%

\begin{lemma} [Chordal Path Decomposition]
\label{lem:scg_decomposition}
Let $G = G[Q]$ be strongly connected. Then there exists a sequence
$G_0 \subseteq \dots \subseteq G_m = G$ of subgraphs of $G$ such that
\begin{packed_enum}
\item $G_0$ is a simple cycle
\item For every $i=1, m$, $G_i = G_{i-1} \cup P_i$, where $P_i$ is a
  chordal path of $G_{i-1}$.
\end{packed_enum}
\end{lemma}

\begin{proof}
  We construct $G_i$ inductively.  Let $G_0$ be any simple cycle in
  $G$ (there exists one, since $G$ is strongly connected).  For $i
  \geq 1$, suppose $G_{i-1} \neq G$.  Since $G$ is strongly connected,
  there exists some edge $e_{R_0} = (u,v) \in E(G) \setminus
  E(G_{i-1})$ such that $u \in V(G_{i-1})$, and there exists a simple
  path $P'_i$ from $v$ to some node in $G_{i-1}$, $P'_i : v \leadsto w$,
  $w \in V(G_{i-1})$ (if $v \in V(G_{i-1})$, then $P'_i$ is empty and
  $w=v$).  Define $P_i = e_{R_0}, P'_i$ and $G_i= G_{i-1} \cup P_i$.
%
%
\end{proof}

\begin{example} \label{ex:h2}
  Consider the query $H_2 = R(\ux,y), S(\uy, z), T(\uz, x), U(\uy, t),
  V(\underline{t}, z)$.
The query admits the following  decomposition:
\begin{align*}
  G_0 = & G[Q_0] & \mbox{where } Q_0 = R(\ux,y), S(\uy, z), T(\uz, x) \\
  G_1 = & G_0 \cup P & \mbox{where } P = U(\uy, t), V(\underline{t}, z)
\end{align*}
%
\end{example}


Our algorithm for computing \cqa{$Q$} for an SCC $Q$ uses a chordal
path decomposition of $Q$ and applies the following two procedures.

\introparagraph{Procedure \textsc{FrugalChord}} Fix a query $Q$ of the form
$Q_0, R^c(\underline{u},v)$, where $R^c(\underline{u},v)$ is a chord
for $Q_0$.  The procedure \textsc{FrugalChord} takes as input an
instance $I$ and the compact representation $A_{Q_0}(I)$, and returns
the compact representation $A_Q(I)$.  The procedure simply returns
the set:
\begin{align}
A_Q(I) = \setof{A \in A_{Q_0}(I)}{\forall t \in A: (t[u], t[v]) \in R^c} \label{eq:chord}
\end{align}
In other words, the procedure computes a representation of $Q$ on $I$
by having access to a representation to $Q_0$ on $I$.  Correctness
follows from:

\begin{lemma}
\label{lem:cycle_edge_representable}
Let $Q \equiv Q_0, R^c(\underline{u},v)$ such that
$R^c(\underline{u},v)$ is a chord of $Q_0$.  If $\mM_{Q_0}(I)$ is
representable and its compression is $A_{Q_0}(I)$, then ,
$\mM_{G_{i+1}}(I)$ is also representable and its compression is given
by Eq.(\ref{eq:chord}).
\end{lemma}

\begin{proof}
For the one direction, consider a frugal repair $r$ with answer set $Q^f(r)$. We need to show
that for any tuple $t \in Q^f(r)$, $t \in O_A$ for some or-set $O_A \in A_{Q_0}(I)$ such that
for all tuples $t' \in O_A$, $(t'[u], t'[v]) \in R^c$. Indeed, let $t' \in O_A$ be a tuple for which 
$(t'[u], t'[v]) \notin R^c$. Then, we can create a repair $r'$ that returns a strictly smaller answer set
than $r$ (does not include $t$). For the other direction, let $\{ t_1, \dots , t_m \} \in \alpha(B)$, where
$B$ is the R.H.S. of Eq.(\ref{eq:chord}), and $r$ repair such that $Q^f(r) = \{t_1, \dots, t_m\}$. Then,
$r$ must be frugal, otherwise we would have a contradiction on the fact that $A_{Q_0}(I)$ is correctly
structured.
\end{proof}

\introparagraph{Procedure \textsc{FrugalChordPath}} Fix a query $Q$ of the form
$Q_0, P$, where $P$ is a chordal path from $u$ to $v$ for $Q_0$.  The
procedure \textsc{FrugalChordPath} takes as input an instance $I$ and
the compact representation $A_{Q_0}(I)$, and returns the compact
representation $A_Q(I)$, in six steps:

\begin{packed_enum}
\item \label{item:proc:chordpath:1} Assume $A_{Q_0}(I)$ has $m$
  or-sets, each with $n_1, \ldots, n_m$ elements:
  \begin{align}
    A_{Q_0}(I) = \set{A_1, \ldots, A_m} 
\mbox{ where: } A_i =  \orset{t_{i1}, t_{i2}, \ldots, t_{in_i}} \label{eq:renaming} 
  \end{align}
  Denote $n=\sum_i n_i$.  Let $a_i$ for $i=1,m$ be $m$ distinct
  constants, and let $b_{ij}$ for $i=1,m$, $j=1,n_i$ be $n$ distinct
  constants.  Denote $\texttt{tup}(b_{ij}) = t_{ij}$ the tuple encoded
  by $b_{ij}$.
\item \label{item:proc:chordpath:2} Create four new relations:
  \begin{align*}
    B^i = & \setof{(\underline{a_i}, b_{ij})}{i=1,m; j=1,n_i}\\
    B_1^c = & \setof{(\underline{b_{ij}},\pi_u(t_{ij}))}{i=1,m; j=1,n_i}\\
    B_2^c = & \setof{(\underline{b_{ij}},\pi_v(t_{ij}))}{i=1,m; j=1,n_i}\\
    B_0^c = & \setof{(\underline{\pi_v(t_{ij})},a_i)}{i=1,m; j=1,n_i}
  \end{align*}
  $B^i$ is of inconsistent type (hence the superscript ``$i$''),
  and $B_1^c$, $B_2^c$, $B_0^c$ are of consistent type.
\item \label{item:proc:chordpath:3} Assume the variables $u,v$ are
  distinct, $u\neq v$: we discuss below the case $u=v$.  Denote
  $C_{k+3}$ and $Q'$ the following queries:
  \begin{align*}
    &C_{k+3} = B^i(\underline{a},b),B_1^c(\underline{b},u),R_1(\underline{u},x_1),\ldots,R_k(\underline{x_{k-1}},v),B_0^c(\underline{v},a)\\
    &Q' = C_{k+3}^f(a,b,u,x_1,\ldots,x_{k-1},v), B_2^c(\underline{b},v)
  \end{align*}
  where $R_1(\underline{u},x_1),\ldots,R_k(\underline{x_{k-1}},v)$ is the chordal path $P$,
  and $a, b$ are new variables.
\item \label{item:proc:chordpath:4} Use the algorithm \textsc{FrugalC}
  to find the compact representation $A_{C_{k+3}}(I)$ for $C_{k+3}$.
\item \label{item:proc:chordpath:5} Use the procedure
  \textsc{FrugalChord} to find the compact representation of
  $A_{Q'}(I)$ for $Q'$.
\item \label{item:proc:chordpath:6} Return the following set of
  or-sets:
  \begin{align}
    A_Q(I) = &
    \setof{\orset{(\texttt{tup}(\pi_b(t)),\pi_{Vars(P)}(t))|t \in A}}{A \in A_{Q'}(I)} \label{eq:aq:chordal:path}
  \end{align}
\end{packed_enum}

We explain the algorithm next.  In Step~\ref{item:proc:chordpath:1} we
give fresh names to each or-set $A_i$ in $A_{Q_0}(I)$, and to each
tuple $t_{ij}$ in each or-set in $A_i$: by
\autoref{prop:compression:size}, the number of constants needed is
only polynomial in the size of the active domain of $I$.  The crux of
the algorithm is the table $B^i(\underline{a},b)$ created in
Step~\ref{item:proc:chordpath:2}: its repairs correspond precisely to
$\alpha(A_{Q_0}(I))$, up to renaming of constants.  To see this notice
that each repair of $B^i$ has the form $\set{(a_1,b_{1j_1}), \ldots,
  (a_m,b_{mj_m})}$ for arbitrary choices $j_1 \in [n_1], \ldots, j_m
\in [n_m]$.  Therefore, the set of frugal repairs of $B^i$ is $\alpha(S_0)$,
where $S_0 = \setof{\orset{(a_i,b_{ij})|j=1,n_i}}{i=1,m}$, which is
precisely Eq.(\ref{eq:renaming}) up to renaming of the tuples by
constants.  The relation $B_1^c$ decodes each constant $b_{ij}$ by
mapping it to the $u$-projection of $t_{ij}$; similarly for
$B_2^c$. Clearly, both $B_1^c,B_2^c$ are consistent, because every
constant $b_{ij}$ needs to be stored only once.  The relation $B_0^c$
is a reverse mapping, which associates to each value of $v$ the name
$a_i$ of the unique or-set $A_i$ that contains a tuple $t_{ij}$ with
that value in position $v$: the set $A_i$ is uniquely defined because,
by \autoref{def:representable}, for any distinct sets $A_{i_1},
A_{i_2}$ we have $\Pi_v(A_{i_1}) \cap \Pi_v(A_{i_2})=\emptyset$.

Step~\ref{item:proc:chordpath:3} transforms $Q$ into a cycle $C_{k+3}$
plus a chord $B^c_2(\underline{b},v)$, by simply replacing the entire
subquery $Q_0$ with the single relation $B^i(\underline{a},b)$ (which
is correct, since $A_{Q_0}(I)$ is the same as the set of repairs of
$B^i$) plus the decodings $B_1^c(\underline{b},u)$,
$B_2^c(\underline{b},v)$: note that we only needed
$B^c_0(\underline{v},a)$ in order to close the cycle $C_{k+3}$.  The
next two steps compute the encodings $A_{C_{k+3}}(I)$ and $A_{Q'}(I)$
using the algorithm \textsc{FrugalC} and \textsc{FrugalChord}
respectively.  Finally, the last step converts back $A_{Q'}(I)$ into
$A_Q(I)$ by expanding the constants $b_{ij}$ into the tuples they
encode, $t_{ij} = \texttt{tup}(b_{ij})$.  The algorithm has assumed $u
\neq v$.  If $u=v$ are the same variable, the $C_{k+3}$ is no longer a
cycle: in that case, we split $u$ into two variables $u,v$ and add two
consistent relations $R^c(\underline{u},v)$, $S^c(\underline{v},u)$ to
the query, and replace the last relation $R_k(\underline{x_{k-1}},u)$
of $P$ with $R_k(\underline{x_{k-1}},v)$. The correctness of the
algorithm follows from:

\begin{lemma}
  Let $Q$ be a query of the form $Q_0, P$ where $P$ is a chordal path
  from $u$ to $v$ for $Q_0$, and let $I$ be an instance.  Then, if
  $\mM_{Q_0}(I)$ is representable and $A_{Q_0}(I)$ is its compact
  representation, then $\mM_Q(I)$ is also representable and its
  compact representation is given by Eq.(\ref{eq:aq:chordal:path}).
\end{lemma}

\introparagraph{Algorithm \textsc{FrugalSCC}} Let $Q$ be a query that is
strongly connected.  The algorithm \textsc{FrugalSCC} takes as input
an instance $I$, and returns $A_Q(I)$, as follows.  Let $Q_0, Q_1,
\ldots, Q_m$ a chordal path decomposition for $Q$
(\autoref{lem:cycle_edge_representable}).  Start by computing
$A_{Q_0}(I)$ using algorithm \textsc{FrugalC}.  Next, for each
$i=1,m$, use $A_{Q_{i-1}}(I)$ and the procedure
\textsc{FrugalChordalPath} to compute $A_{Q_i}(I)$.  Return
$A_{Q_m}(I)$.

\begin{example} Continuing \autoref{ex:h2}, we will show how to
  compute $A_{H_2}(I_2)$ where $I_2$ is the instance shown in
  \autoref{fig:h2_instance}.  Write $H_2$ as $H_2 \equiv C_3, P$,
  where $C_3=R(\ux,y), S(\uy, z), T(\uz, x)$ and $P=U(\uy, t),
  V(\underline{t}, z)$.  We start by computing $C_3$ on $I_2$; one can
  check\footnote{Every repair of $I_2$ contains exactly two cycles:
    $(a_1,b_1,c_1)$ and one of $(a_2,b_2,c_2)$ or $(a_2,b_3,c_2)$.}
  that $A_{C_3}(I_2) = \set{A_1, A_2}$ where $A_1 = \orset{(a_1, b_1,
    c_1)}$ and $A_2 = \orset{(a_2, b_2, c_2), (a_2, b_3, c_2)}$.
  Encode the two sets with the new constants $A_1, A_2$, and encode
  the three tuples with three new constants $[a_1b_1c_1]$,
  $[a_2b_2c_2]$, $[a_2b_3c_2]$.  The new relations
  $B^i(\underline{a},b)$, $B^c_1(\underline{b},y)$,
  $B^c_2(\underline{b},z), B_0^c(\uz, b)$ are shown in
  \autoref{fig:scc_reduction}.  Thus, we have to compute the following
  queries:
  \begin{align*}
    C_5 = &  B^i(\underline{a},b),B_1^c(\underline{b},y),U_1(\uy,t),V(\underline{t},z),B_0^c(\underline{z},a)\\
    Q' = & C_5^f(a,b,y,t,z), B_2^c(\underline{b},z)
  \end{align*}
  on the instance $I'$ in~\autoref{fig:scc_reduction}. One can check
  that their answers are:
  \begin{align*}
    &A_{C_5}(I') = \set{\orset{(A_1, [a_1b_1c_1], b_1, d, c_1),   (A_2, [a_2b_2c_2], b_2, d, c_2),  (A_2, [a_2b_3c_2], b_3, d, c_2) }}\\
    &A_{Q'}(I') = A_{C_5}(I')
  \end{align*}
  Mapping this to the original query $H_2(x,y,z,t)$ by projecting out
  the $A_i$ and merging the tuples, we obtain that
  \begin{align*}
    A_{H_2}(I_2) = \set{ \orset{(a_1, b_1, c_1, d), (a_2, b_2, c_2, d), (a_2, b_3, c_2, d)}}
\end{align*}
In particular, $I_2 \vDash H_2$, because $A_{H_2}(I_2)$ is nonempty.
\end{example}

\begin{figure}[tb]
  \centering
  {{ \renewcommand\tabcolsep{5pt}
        \begin{tabular}{|c|c|c|c|c|c|c|c|c| }
            \multicolumn{1}{l}{$R(\ux,y)$} & \multicolumn{1}{l}{ } & 
            \multicolumn{1}{l}{$S(\uy, z)$} & \multicolumn{1}{l}{ } & 
            \multicolumn{1}{l}{$T(\uz,x)$}  & \multicolumn{1}{l}{ } &
            \multicolumn{1}{l}{$U(\uy,t)$} & \multicolumn{1}{l}{ } & 
            \multicolumn{1}{l}{$V(\underline{t}, z)$}   \\
            
            \cline{1-1}\cline{3-3} \cline{5-5}  \cline{7-7} \cline{9-9}                                         
             $(a_1, b_1)$ &  & $(b_1, c_1)$ &  & $(c_1, a_1)$ &  & $(b_1, d)$ & & $(d, c_1)$ \\
             $(a_2, b_2)$ &  & $(b_2, c_2)$ &  & $(c_2, a_2)$ &  & $(b_2, d)$ & & $(d, c_2)$ \\
             $(a_2, b_3)$ &  & $(b_3, c_2)$ &  &              &  & $(b_3, d)$ & &  \\
             \cline{1-1}       \cline{3-3}       \cline{5-5}        \cline{7-7} \cline{9-9}
          \end{tabular}}}
\caption{{\footnotesize An inconsistent purified instance $I_2$ for $H_2$.}}
\label{fig:h2_instance}
\end{figure}

\begin{figure}[tb]
  \centering
  {{ \renewcommand\tabcolsep{5pt}
        \begin{tabular}{|c|c|c|c|c|c|c| }
            \multicolumn{1}{l}{$B(\underline{a}, b)$} & \multicolumn{1}{l}{ } & 
            \multicolumn{1}{l}{$B_1^c(\underline{b}, y)$}  & \multicolumn{1}{l}{ } &
            \multicolumn{1}{l}{$B_2^c(\underline{b},z)$}  & \multicolumn{1}{l}{ } & 
            \multicolumn{1}{l}{$B_0^c(\uz, b)$}  \\
            
            \cline{1-1}\cline{3-3} \cline{5-5}  \cline{7-7}                                         
           $(A_1, [a_1b_1c_1])$ &  & $([a_1b_1c_1], b_1)$ & & $([a_1b_1c_1], c_1)$ & & $(c_1, A_1)$ \\
           $(A_2, [a_2b_2c_2])$ &  & $([a_2b_2c_2], b_2)$ & & $([a_2b_2c_2], c_2)$ & & $(c_2, A_2)$ \\
           $(A_2, [a_2b_3c_2])$ &  & $([a_2b_3c_2], b_3)$ & & $([a_2b_3c_2], c_2)$ & &  \\
            \cline{1-1}\cline{3-3} \cline{5-5}   \cline{7-7}
          \end{tabular}
  }}
\caption{{\footnotesize The resulting instance $I'$ produced by the inductive step for $H_2$.}}
\label{fig:scc_reduction}
\end{figure}

\section{The PTIME algorithm}
\label{sub:ptime_case}

In this section, we prove:

\begin{theorem}
If the graph $G[Q]$ is splittable, there exists a PTIME algorithm that solves \cqa{$Q$}.
\end{theorem}

The polynomial time algorithm we present here is based on the fact that if $G[Q]$ is 
splittable, it has a very specific structure that allows us to break it into smaller pieces
that we can solve independently; in other words, the problem is {\em self-reducible}. 
The graph object that allows this is called a {\em separator}, and we show in \autoref{sub:separator}
that it always exists in $G[Q]$. Throughout this section, we will use the graph $G[H]$ of
\autoref{fig:q_graph} as a running example.

\subsection{Separators}
\label{sub:split:sets}

In this section, we define the notion of a {\em separator}, which is central
to the construction of the polynomial time algorithm for splittable graphs. Before we present the
formal definition, we need to set up some notation.

Recall that $\sim$ denotes a binary relation between edges $R,S \in E^i$: $R \sim S$ if 
$R$ and $S$ are source-equivalent. Consider the equivalence relation defined by $\sim$ 
on the set of inconsistent edges $E^i$, and denote $\Eq$ the quotient set and $[R] \in \Eq$ the 
equivalence class for an edge $R \in E^i$. For our example graph $G[H]$, we have 
$R_1 \sim R_3$ (because $R_1, R_2,  R_3$ form a cycle), thus $[R_1] = \set{R_1, R_3}$.  
Also $S \lesssim  [R_1]$, $S \lesssim T$,  hence $\Eq = \set{[R_1], [S], [T]}$.

For some $C \in \Eq$, let us define
\begin{align*}
C^+  \stackrel{def}{\equiv}  \bigcap_{R \in C} u_R^{+,R} 
\quad \text{ and } \quad 
C^{\oplus}  \stackrel{def}{\equiv}  \bigcap_{R \in C} u_R^{\oplus}.
\end{align*}

Similarly to how we have defined $coupled^+(R), coupled^{\oplus}(R)$ for edges $R \in E^i$, we can
define $coupled^+(C), coupled^{\oplus}(C)$ for $C \in \Eq$:
\begin{align*}
coupled^+(C)  & \stackrel{def}{\equiv}  \set{C} \cup \{ C' \in \Eq  \mid \exists R \in C, S \in C': 
 \exists P: v_R \leftrightarrow u_S,  P \cap C^+ = \emptyset \} \\
 coupled^{\oplus}(C) & \stackrel{def}{\equiv}  \set{C} \cup \{ C' \in \Eq \mid  \exists R \in C, S \in C': 
 \exists P: v_R \leftrightarrow u_S, P \cap C^{\oplus} = \emptyset \}
\end{align*}

The definitions essentially "lift" the notion of coupling from a single inconsistent edge to an equivalence class. To illustrate with an example, in $G[H]$ we have the following:
\begin{align*}
& coupled^+(\set{R_1, R_3}) = \set{\set{R_1, R_3}, \set{S}} \quad \quad coupled^+(\set{S}) = \set{\set{S}} \\
& coupled^+(\set{T}) = \set{\set{R_1, R_3}, \set{T}, \set{S}}
\end{align*}
Moreover, for every equivalence class in $G[H]$, the sets $coupled^+, coupled^{\oplus}$ coincide.

For $C_1, C_2 \in \Eq$, define the binary relation $\leq^{\oplus}$: we say that
$C_1 \leq^{\oplus} C_2$ if there exists some $S \in C_2$ such that $u_S \in C_1^{\oplus}$.

\begin{proposition}
\label{lem:strict_porder}
The relation $\leq^{\oplus}$ is antisymmetric and transitive.
\end{proposition}

\begin{proof} 
To show that  $\leq^{\oplus}$ is {\em antisymmetric}, notice that if 
$C_1 \leq^{\oplus} C_2$ and $C_2 \leq^{\oplus} C_1$, $C_1$ and $C_2$ would describe the same 
equivalence class, and thus $C_1 = C_2$. To show {\em transitivity}, assume that 
$C_1 \leq^{\oplus} C_2$ and $C_2 \leq^{\oplus} C_3$. Then, there exists $S \in C_2$ such that
$u_S \in C_1^{\oplus}$ and also $T \in C_3$ such that  $u_T \in C_2^{\oplus}$, and
in particular $u_T \in u_S^{\oplus}$. Thus, $u_T \in C_1^{\oplus}$ and $C_1 \leq^{\oplus} C_3$.
\end{proof}

We can now define $C_1 <^{\oplus} C_2$ to be such that $C_1 \leq^{\oplus} C_2$ and
$C_1 \neq C_2$. Then, following from \autoref{lem:strict_porder}, $<^{\oplus}$ is a {\em strict
partial order}. We will be particularly interested in the maximal elements of this order,
which we will call {\em sinks}. 

\begin{definition}[Sink]
$C \in \Eq$ is a {\em sink} if it is a maximal element of $<^{\oplus}$.  
\end{definition}

\begin{example}
Since $(u_{R_3} = ) z \in u^{\oplus} (= u_S^{\oplus}) $, we have $\set{S} <^{\oplus} \set{R_1, R_3}$.
Also, since $v \in u_{R_1}^{\oplus} \cap u_{R_3}^{\oplus}$, $\set{R_1, R_3} <^{\oplus} \set{T}$.
By the transitivity of $<^{\oplus}$, we also obtain that $\set{S} <^{\oplus} \set{T}$.
Hence, $\set{T}$ is the only sink of the graph $G[H]$.
\end{example}

\begin{definition}[Separator]
A sink $C \in \Eq$ is a {\em separator} if for every $C' \neq C$ such that
$C' \in coupled^{\oplus}(C)$, we have that $C' <^{\oplus} C$.
\end{definition}

In the specific case where $\Eq$ contains a single sink $C$, since $<^{\oplus}$ is a strict partial order, for any $C' \in \Eq, C' \neq C$, we have that $C' <^{\oplus} C$ and thus the single sink $C$ is trivially a separator.

All the equivalence classes of $G[H]$ are separators. Indeed, since
$\set{S} <^{\oplus} \set{R_1, R_3} <^{\oplus} \set{T}$, $\set{T}$ is a separator.
Also, $\set{S}$ is a separator, since $\set{R_1, R_3}, \set{T} \notin coupled^{\oplus}(S)$.

In order to prove the existence of a separator in a graph, it is not a sufficient condition that the graph is splittable. For example, consider the splittable query $Q = R^i(\ux, y), S^i(\ux, y), T^i(\uz,y)$, which contains two sinks, $\set{R,S}$ and $\set{T}$. It is easy to see that $\set{T} \notin coupled^{\oplus}(\set{R,S})$, and $\set{R,S} \notin coupled^{\oplus}(\set{T})$; thus, $G[Q]$ has no separator. Instead, we show the existence of a separator for a graph that is splittable and {\em f-closed}. 

\begin{definition}[f-closed Graph]
A graph $G$ is {\em f-closed} if for every $R \in E^i(G)$,
$v_R^{\oplus} \cap u_R^{+,R} \subseteq u_R^{\oplus}$.
\end{definition}

Indeed, $G[Q]$ is not f-closed, since $v_R^{\oplus} = \set{y}$, $u_R^{+,R} = \set{y}$ and $u_R^{\oplus} = \set{x}$. We will show in \autoref{sub:f-closed} that, given a splittable graph $G$ and an instance $I$, we can always construct in polynomial time a splittable and f-closed graph $G'$ and an instance $I'$ such that $I \vDash G$ iff $I' \vDash G'$.

We show in \autoref{sub:separator} that, if $G$ is splittable and f-closed, there exists a separator, and in fact the separator has an explicit construction:

\begin{theorem}
\label{thm:separator_exists}
If $G$ is a splittable and f-closed graph, then
$C^{sep} = \arg \min_{\text{sink } C \in \Eq} |coupled^{\oplus}(C)|$ is a separator.
\end{theorem}

In other words, the sink $C$ with the smallest $coupled^{\oplus}(C)$ is a separator (there can be many). In the next subsection, we use the existence of a separator to design a recursive polynomial time algorithm for splittable graphs.

\subsection{The Recursive Algorithm}
\label{sec:recursive:algo}

We present here an algorithm, $\textsc{RecursiveSplit}$, that takes as input an instance $I$
and a splittable graph $G$ and returns \texttt{True} if $I \vDash G$, otherwise \texttt{False}. 
The algorithm is recursive on the number of inconsistent relations, $|E^i(G)|$ of $G$.
For the base case $E^i(G) = \emptyset$ (all relations are consistent), we have
that $\textsc{RecursiveSplit}(I,G) = \texttt{True}$ if and only if $G(I)$ is true.

We next show how to recursively compute $\textsc{RecursiveSplit}(I,G)$ when $|E^i(G)| > 0$. Since $G$ is a splittable and f-closed graph, \autoref{thm:separator_exists} tells us that there exists a separator $C$. We partition the edges of $E^i$ into a left ($\mL$) and right ($\mR$) set  as follows:
\begin{align*}
\mL^C  = \setof{R \in E^i}{[R] \in coupled^{\oplus}(C)} \quad , \quad 
\mR^C  = E^i \setminus \mL^C
\end{align*}
Let $S_C$ denote the unique SCC that contains all the sources for the edges in $C$. Recall from \autoref{sec:cycles} that one can use the algorithm \textsc{FrugalSCC} to compute  the compression $A_{S_C}(I)$ of $\mM_{S_C}(I)$ in polynomial time, since $S_C$ is a strongly connected graph. Let $\mathcal{A}$ denote the set of all tuples that appear in some or-set of $A_{S_C}(I)$, and $\mathcal{B} = \Pi_{C^{\oplus}}(G^f(I))$. For some $\mba \in \mathcal{A}$, we say that $\mba$ is {\em aligned} with $\mbb \in \mathcal{B}$, denoted $\mba \| \mbb$, if there exists a tuple $t \in G^f(I)$ such that $t[V(S_C)] = \mba$ and $t[C^{\oplus}] = \mbb$. Also, define $algn(\mbb) = \setof{\mba \in \mathcal{A}}{\mba \| \mbb}$. Observe that $\mba$ can be aligned with at most one $\mbb$, since there exists a consistent directed path from every node of $V(S_C)$ to every node of $C^{\oplus}$. Notice also that when $C^{\oplus} = \emptyset$, all the tuples in $\mathcal{A}$ are vacuously aligned with the empty tuple $()$. 

For every $\mbb \in \mathcal{B}$, choose a tuple $t^{(\mbb)} \in G^f(I)$ such that $t^{(\mbb)}[C^{\oplus}] = \mbb$. For every tuple $\mba \in \mathcal{A}$, we now define a subinstance $I[\mathbf{a}] \subseteq I$ such that:
\begin{align*}
R^{I[\mba]} = 
\begin{cases}
R^{I}  & \text{if $R \in E^c(G)$},\\
\setof{ (t^{(\mbb)}[u_R], t^{(\mbb)}[v_R])}{\mbb: \mba \| \mbb} & \text{if $R \in \mathcal{R}^C$},\\
\setof{(t[u_R], t[v_R])}{t \in G^{f}(I) \mbox{ s.t. } t[V(S_C)] = \mba} & \text{if $R \in \mL^C$}.
\end{cases}
\end{align*}

Notice that if some relation $R$ belongs in $S_C$, then it must contain exactly one tuple, while if $u_R$ belongs in $V(S_C)$, then $R^{I[\mba]}$ contains exactly one key-group. On the other hand, the relations that do not belong in LC contain only one tuple that contributes to $t^{(\mbb)}$. 

The first key idea behind the above construction of subinstances is captured by the following lemma, which shows that certain subinstances are independent in the relations of $\mL^C$.

\begin{lemma} 
\label{lem:key-group:independence}
Let $\mba_1, \mba_2 \in \mathcal{A}$. The instances $I[\mathbf{a}_1], I[\mathbf{a}_2]$ share no key-groups in any relation $R \in \mL^C$ if either of the following two conditions hold:
\begin{packed_enum}
\item $\mba_1, \mba_2$ belong in different or-sets of $A_{S_C}(I)$.
 \item $\mba_1 \| \mbb_1, \mba_2 \| \mbb_2$, and $\mbb_1 \neq \mbb_2$.
\end{packed_enum}
\end{lemma}

\begin{proof}
To show (1), assume for some $R \in \mL^C$ that the key-group $R(\underline{c},-)$ appears both in $I[\mathbf{a}_1],  I[\mathbf{a}_2]$. Since $[R] <^{\oplus} C$, there exists a path $P_S: u_R \rightarrow u_S$, where $u_S \in V(S_C)$. It follows from our construction that  both $(c,\mathbf{a}_1[u_S]), (c,\mathbf{a}_2[u_S]) \in \Pi_{u_R, u_S}(P_S^f(I))$. But since $P_S$ contains only consistent relations, it must be that $\mathbf{a}_1[u_S] = \mathbf{a}_2[u_S]$, a contradiction to the fact that $\mathbf{a}_1, \mathbf{a}_2$ are value-disjoint (since they belong in different or-sets).

To show (2), let $R \in \mL^C$ and assume that a key-group $R(\underline{c},-)$ appears both in $I[\mathbf{a}_1],  I[\mathbf{a}_2]$. As the argument for (1), there will be some $u_S \in V(S_C)$ such that $\mba_1[u_S] = \mba_2[u_S]$. Since $\mba_1 \| \mbb_1$, there exists a tuple $t_1 \in G^f(I)$ such that $t_1[u_S] = \mba_1[u_S]$ and $t_1[C^{\oplus}] = \mbb_1$. Similarly, since $\mba_2 \| \mbb_2$, there exists a tuple $t_2 \in G^f(I)$ such that $t_2[u_S] = \mba_2[u_S]$ and $t_2[C^{\oplus}] = \mbb_2$. But now, $t_1[u_S] = t_2[u_S]$ and $t_1[C^{\oplus}] \neq t_2[C^{\oplus}]$, which is a contradiction, since each node $u_S$ for $S \in C$ has a consistent path $P: u_S \rightarrow v$ for every $v \in C^{\oplus}$. 
\end{proof}

The second key idea is that computing whether $I[\mba] \vDash G$ can be reduced to a computation where $G$ contains strictly less inconsistent relations. Indeed, recall that in $I[\mba]$, every relation $R_i \in C$, $i=1, \dots, m$, contains exactly one key-group, $R_i(\underline{\mathbf{a}[u_{R_i}]},-)$ (and if it both vertices of $R$ are in $S_C$, it contains exactly one tuple). We can now apply a "brute force" approach and try all the possible combinations of choices for these key-groups, since they are polynomially many: each such combination will create a new instance where the relations in $C$ will be consistent, and thus can be computed in polynomial time by induction. The procedure $\textsc{Simplify}(I[\mba],G)$ formally presents the algorithm we sketched.

\begin{algorithm}
\SetAlgoLined
\DontPrintSemicolon
$K = \setof{(c_1, \dots, c_m)}{\forall i: R_i(\underline{\mba[u_{R_i}]}, c_i)}$ \;
$G' \leftarrow G$ where all edges of $C$ are of consistent type \; 
$\forall \mathbf{c} \in K$:  
$I[\mba]^{(\mathbf{c})} \leftarrow (I[\mathbf{a}] \setminus \bigcup_{i=1}^m R_i(\mba[u_{R_i}],-))  \bigcup_{i=1}^m R_i(\mathbf{a}[u_{R_i}], c_i)$ \;
\KwRet{ $(\forall \mathbf{c} \in K$: 
$\textsc{RecursiveSplit}(I[\mba]^{(\mathbf{c})},G') = \texttt{True})$}
\caption{\textsc{Simplify}($I[\mba],G$)}
\end{algorithm}

\begin{algorithm}
\SetAlgoLined
\DontPrintSemicolon
\lIf{$E^i(G) = \emptyset$} {\KwRet{G(I)}}
Find a separator $C$ of $G$ \;
$\mathcal{B} \leftarrow \Pi_{C^{\oplus}}(G^f(I))$ \;
$A_{S_C}(I) \leftarrow \textsc{FrugalSCC}(I, S_C)$ \;
\For{ $ \mathbf{b} \in \mathcal{B}$}{
  \uIf{ $\exists$ or-set $A \in A_{S_C}(I)$ s.t. $\forall \mba \in A \cap algn(\mbb) \Rightarrow \textsc{Simplify}(I[\mba],G) = \texttt{True}$} 
  {
 $r[\mathbf{b}] \leftarrow$ any repair of 
   $\left (\bigcup_{\mba \in algn(\mbb)} I[\mba] \right)$
  }
   \uElse {$r[\mathbf{b}] \leftarrow \emptyset$}
}
$\forall R \in E(G)$: $R^{I'} =\begin{cases}
R^{I} \cap ( \bigcup_{ \mathbf{b} \in \mathcal{B}} r[\mathbf{b}] ) & \text{if $R \in \mL^C$},\\
R^I& \text{otherwise}.
\end{cases}$ \;
$G' \leftarrow G$ where all edges in $\mL^C$ are of consistent type\;
\KwRet{$\textsc{RecursiveSplit}(I', G')$}
\caption{\textsc{RecursiveSplit}($I,G$)}
\end{algorithm}

We first argue that the algorithm \textsc{RecursiveSplit} runs in polynomial time. First,  the final recursive call on $I,G'$, the graph $G'$ has $|E^i(G)| - |\mL^C| < |E^i(G)|$ inconsistent edges, so by the induction argument can be computed in polynomial time. Second, the algorithm calls $\textsc{Simplify}(I[\mba], G)$ at most $|\mathcal{A}|$ times, and we have shown that each such call can be computed in polynomial time. We next argue that \textsc{RecursiveSplit} correctly computes whether $I \vDash G$ or not. We prove first:

\begin{lemma}
\label{lem:equiv:computation}
$\bigcup_{\mba \in algn(\mbb)} I[\mba] \vDash G$ if and only if  there exists an or-set $A \in A_{S_C}(I)$ such that for every $\mba \in A \cap algn(\mbb)$, $I[\mba] \vDash G$.
\end{lemma}

\begin{proof}
For the one direction, assume for the sake of contradiction that for every or-set $A_i$ (let $i=1, \dots, M$), there exists a tuple $\mba_i \in A_i$ such that $\mba_i \| \mbb$ and $I[\mba_i] \not \vDash G$. Then,
there exists a repair $r^{(\mba_i)} \subseteq I[\mba_i]$ such that $G(r^{(\mba_i)})$ is false. By \autoref{lem:key-group:independence}(1), the repairs $r^{(\mba_i)}$  will never conflict on their choices for the relations in $\mL^C$, and by the construction of $I[\mba_i]$, all the other relations are consistent and contain the same tuples. Hence, $r = \bigcup_{\mba \in algn(\mbb)} r^{(\mba_i)}$ is a repair for $\bigcup_{\mba \in algn(\mbb)} I[\mba]$ that cannot satisfy $G$, a contradiction.

For the inverse direction, assume that there exists an or-set $A$ such that for every $\mba \in A \cap algn(\mbb)$, $I[\mba] \vDash G$. If $r$ is a frugal repair for $\bigcup_{\mba \in algn(\mbb)} I[\mba]$, it must be that $\mba \in \prod_{V(S_C)} G^f(r)$ for some $\mba \in A$. But since $I[\mba] \vDash G$, any choice that $r$ has made on the key-groups of $\mL^C$ that appear in $I[\mba]$ will create a tuple in $G^f(r)$.
\end{proof}

Given a repair $r$ of $I$ and a repair $r^{(\mbb)}$ of $\bigcup_{\mba \in algn(\mbb)} I[\mba]$, we define $merge^{C}(r, r^{(\mbb)})$ as a new repair $r_m$ of $I$ such that for any key-group $R(\underline{a},-)$, if $R \notin \mL^C$ or $r^{(\mbb)}$ does not contain the key-group, $r_m$ includes the choice of $r$; otherwise, it includes the choice of $r^{(\mbb)}$. In other words, to construct $r_m$ we let $r^{(\mbb)}$ overwrite $r$ only in the relations of $\mL^C$. Our main technical lemma states:

\begin{lemma}
\label{lem:repair:substitute}
For any frugal repair $r$ of  $I$:
\begin{packed_enum}
\item If $\bigcup_{\mba \in algn(\mbb)} I[\mba] \not \vDash G$ then $\mbb \notin \prod_{C^{\oplus}} G^f(r)$.
\item If $\bigcup_{\mba \in algn(\mbb)} I[\mba] \vDash G$ then for any repair $r'$ of  $\bigcup_{\mba \in algn(\mbb)}I[\mba]$, $G(r) = G(merge^C(r, r'))$.
\end{packed_enum}
\end{lemma}

\begin{proof}
To show (2), we will show that if $G(r)$ is true, then for $r'' = merge^C(r,r')$, $G(r'')$ is true as well (this suffices to prove (2), since for each repair $r$, there exists a repair $r''$ of $\bigcup_{\mba \in algn(\mbb)} I[\mba]$ such that $r = merge^C(r,r'')$). Let $t \in G^f(r)$. If $t[C^{\oplus}] \neq \mbb$, then $t \in G^f(r'')$ as well, since the merging of $r,r'$ influences only tuples where $t[C^{\oplus}] = \mbb$. So now assume that $t[C^{\oplus}] = \mbb$.

Define the set of vertices $V^{(C)}$ to contain all the nodes $v \in V(G)$ for which there exists a path $P: v \leftrightarrow u_S$ for some $S \in \mathcal{R}^C$ such that $P \cap C^{\oplus} = \emptyset$. We show first:
\begin{lemma}
For any relation $T \in E^i(G)$, $T \in \mL^C$ iff $u_T \in V^{(C)}$.
\end{lemma}

\begin{proof}
We first show that if $u_T \in V^{(C)}$, $T \notin \mL^C$. Indeed, if $T \in \mL^C$ we would have a path $P': u_T \leftrightarrow v_R$, for some $R \in C$ such that $P' \cap C^{\oplus}  = \emptyset$ and, since $u_T \in V^{(C)}$, another path $P: u_T \leftrightarrow u_S$ for some $S \in \mathcal{R}^C$ where $P \cap C^{\oplus} = \emptyset$. But then, the path $P'' = P', P$ connects $v_R$ with $u_S$ and is not intersected by $C^{\oplus}$, which contradicts the fact that $[S] \notin coupled^{\oplus}(C)$. 

For the other direction, assume that $u_T \notin V^{(C)}$. If $T \notin \mL^C$, then we would have $u_T \in C^{\oplus}$, which would imply that $C <^{\oplus} [T]$. However, this is a contradiction to the fact that $C$ is a separator. 
\end{proof}

For $r'$, there must exist a tuple $t \in G^f(r')$ such that $t[C^{\oplus}] = \mbb$. Now, define a tuple $t''$ as follows: if $v \in V^{(C)}$, $t''[v] = t[v]$, otherwise $t''[v] = t'[v]$. We will show that $t'' \in G^f(r'')$, which proves that $G(r'')$ is true. In particular, we will show that for every relation $T \in E(G)$, the tuple $s = (t''[u_T], t''[v_T])$ belongs in $r''$. We distinguish four cases:
\begin{packed_item}
\item $u_T, v_T \in V^{(C)}$: Then, $s = (t[u_T], t[v_T])$. Clearly, $s$ belongs in $r$, and since $T \notin \mL^C$, $s$ belongs in $r''$ as well.
\item  $u_T, v_T \notin V^{(C)}$: Then, $s = (t'[u_T], t'[v_T])$. Clearly, $s$ belongs in $r'$. If $T$ is consistent, then it will belong in $r''$ as well. If not, then by the above lemma $T \in \mL^C$, which implies that the merging will add $s$ in $r''$.
\item $u_T \in V^{(C)}, v_T \notin V^{(C)}$:  Since there exists a path from $u_T$ to some node $u_S$,
where $S \in \mR^C$, not intersected by $C^{\oplus}$, and no such path from $v_T$, it must be that
$v_T \in C^{\oplus}$. But then, $t''[v_T] = t'[v_T] = \mbb[v_T] = t[v_T]$. Thus, 
$s = (t[u_T], t[v_T])$, and then the argument goes as in the first item.
\item $u_T \notin V^{(C)}, v_T \in V^{(C)}$: this scenario is not possible. Indeed, similar to the above case, it must be that $u_T \in C^{\oplus}$. Now, if $T$ is consistent, we would have $v_T \in V^{(C)}$ as well, a contradiction. If $T$ is inconsistent, then it must be that $T \in C$ (since $C$ is a sink); but then, the fact that $v_T \in V^{(C)}$ implies that $C \in coupled^{\oplus}([R])$, where $R \in \mR^C$, a contradiction. 
\end{packed_item}

To show (1), assume that there exists a tuple $t \in G^f(r)$ such that $t[C^{\oplus}] = \mbb$; we will show that this is a contradiction. Since $\bigcup_{\mba \in algn(\mbb)} I[\mba] \not \vDash G$, there exists a repair $r'$ of $\bigcup_{\mba \in algn(\mbb)} I[\mba] $ such that $G(r')$ is false. Let $r'' = merge^C(r,r')$; we will show that $G^f(r'') \subset G^f(r)$, which contradicts the fact that $r$ is frugal. Notice first that if $t'' \in G^f(r'')$ and $t''[C^{\oplus}] \neq \mbb$, then $t'' \in G^f(r)$ as well.
So now, let $t'' \in G^f(r'')$ such that $t''[C^{\oplus}] = \mbb$. As in the proof for item (1), we construct a tuple $t'$ such that if $v \in V^{(C)}$, $t'[v] = t^{(\mbb)}[v]$, otherwise $t'[v] = t''[v]$, and using a similar argument one can show that $t' \in G^f(r')$; however, this is a contradiction, since $G(r')$ is false.
\end{proof}

To see why \autoref{lem:equiv:computation} and \autoref{lem:repair:substitute} imply the correctness of the algorithm, consider first the case where for some $\mbb \in \mathcal{B}$, for any or-set $A \in A_{S_C}(I)$, there exists some $\mba \in A$ that is aligned with $\mbb$ such that $I[\mba] \not \vDash G$. Then, \autoref{lem:equiv:computation} tells us that $\bigcup_{\mba \in algn(\mbb)} I[\mba] \not \vDash G$ and thus, by \autoref{lem:repair:substitute}(1), for every frugal repair $r$ of $I$, $\mbb \notin \Pi_{C^{\oplus}} G^f(r)$. Hence, all the key-groups of the relations in $\mL^C$ that appear in $I[\mba]$, for any $\mba$ aligned with $\mbb$, can be safely removed from the instance: this is exactly what setting $r[\mbb] = \emptyset$ achieves. On the other hand, assume that  for some $\mbb \in \mathcal{B}$, there exists an or-set $A \in A_{S_C}(I)$, where for every $\mba \in A \cap algn(\mbb)$, $I[\mba] \vDash G$. Then, \autoref{lem:equiv:computation} tells us that  $\bigcup_{\mba \in algn(\mbb)} I[\mba] \vDash G$, and by \autoref{lem:repair:substitute}(2), whether the instance is certain or not is independent of the choice for the key-groups of $\mL^C$ that are contained in $\bigcup_{\mba \in algn(\mbb)} I[\mba]$. 

\subsection{f-closed Graphs}
\label{sub:f-closed}

In this subsection, we show that we can always reduce in polynomial time $G$ with instance $I$ to an f-closed graph $G'$ with instance $I'$ such that  $\mM_G(I) = \mM_{G'}(I')$.
For this, we need the following technical lemma.

\begin{lemma}
\label{lem:closure}
Let $R \in E^i$ and $v \in u_R^{+,R} \cap v_R^{\oplus}$. 
Let $P: u_R, e_R, v_R, \dots, v$ be the directed path from $u_R$ to $v$ with $e_R$ as its first edge. If there exist $(a,b_1), (a,b_2) \in \Pi_{u_R,v}(P^f(I))$
such that $b_1 \neq b_2$, then no frugal repair of $G$ contains $a$.
\end{lemma} 

\begin{proof}
Suppose for the sake of contradiction that there exists a frugal repair $r$ such that for some tuple 
$t \in G^f(r)$, $t[u_R] = a$ and let $t[v] = b$. Assume w.l.o.g. that $b \neq b_1$.
Let us focus on the key-group $R(\underline{a},-)$ and assume that $R(\underline{a},c) \in r$. For 
the tuple $t_P \in P^f(I)$ where $t_P[u_R]=a, t_P[v]=b_1$, it must be that $t_P[v_R] = c' \neq c$ (if
$t_P[v_R] = c$, then it would have been that $t_P[v] = b \neq b_1$).
Now, construct the new repair $r' = (r \setminus \{R(\underline{a},c)\}) \cup \{R(\underline{a},c')\}$.
We will show that $G^f(r') \subsetneq G^f(r)$, which contradicts the frugality of $r$.

First, consider any tuple $t \in G^f(r')$ such that $t[u_R] \neq a$. Then, $t \in G^f(r)$ as well, since $r,r'$
differ only on the choice for the key-group $R(\underline{a},-)$. 
Next, we will show that no tuple $t$ with $t[u_R]=a$
can belong in $G^f(r')$; this completes the proof, since $G^f(r)$ contains such a tuple.
Indeed, in this case we would have $t[v] = b$ (since there exists a directed
path from $u_R$ to $v$ that does not go through $e_R$, which is the only relation where $r,r'$ differ)
and also $t[v] = b_1$ (since now $R(\underline{a},c') \in r'$), which is a contradiction.
\end{proof}

Now, consider some instance $I$ of $G$ such that $G$ is not f-closed. We  present a polynomial
time algorithm, \textsc{F-Closure}, that reduces the graph to an f-closed graph, while
keeping the representation $\mM_G$ the same. Notice that the algorithm has no specific requirements on the structure of $G$.

\begin{algorithm}
\SetAlgoLined
$I_C \leftarrow I, \quad G_C \leftarrow G$\;
\While{ $ \exists R \in E^i(G_C)$, $v \in V(G_C)$ such that 
 $v \in (u_R^{+,R} \cap v_R^{\oplus}) \setminus u_R^{\oplus}$}{
 $P = u_R, e_R, v_R, \dots, v$\;
 $T = \Pi_{u_R,v}(P^f(I))$\;
 $R^{v} = \setof{(a, b) \in T}{ \nexists (a,b') \in T \text{ where } b' \neq b}$\;
 $I_C \leftarrow I_C \cup \{R^v\}$\;
 $G_C \leftarrow (V(G_C), E(G_C) \cup \{(u_R, v)\})$
}
\KwRet{$I_C, G_C$}
\caption{$\textsc{F-Closure}(I,G)$}
\end{algorithm}

\begin{proposition}
\label{prop:f-close:analysis}
Let $I$ be an instance of graph $G$.  \textsc{F-Closure} returns an instance $I_C$
of an f-closed graph $G_C$ in polynomial time such that 
$\mM_G(I) = \mM_{G_C}(I_C)$.
\end{proposition}

\begin{proof}
Note that at an iteration where $v \in (u_R^{+,R} \cap v_R^{\oplus}) \setminus u_R^{\oplus}$, we add
a consistent edge $(u_R,v)$ (such that $v \in u_R^{\oplus}$ in the new graph). Since there are at most
$|E^i(G)| \cdot |V(G)|$ pairs of inconsistent edges and nodes, the algorithm will terminate after that
many steps and return an f-closed graph. It remains to show that if we have the instance $I$ of $G$
at the beginning of the iteration and $I_C, G_C$ at the end, then $\mM_G(I) = \mM_{G_C}(I_C)$.

Notice that there exists a 1-1 correspondence between the repairs of $I, I_C$, since the added relation $R^v$ is consistent.  Let $r_C$ be a repair of $I_C$ and $r$  the corresponding repair of $I$; we will  first show that, if $r$ is frugal, $G^f(r) = G^f_C(r_C)$.
Indeed, $G^f(r) \supseteq G^f_C(r_C)$, since $G_C$ contains additional constraints ($R^v$). To show that $G^f(r) \subseteq G^f_C(r_C)$, let $t \in G^f(r)$, where $t[u_R] = a$. Since $r$ is frugal, by \autoref{lem:closure} there exists a tuple $R^v(\underline{a},b)$, where $t[v]=b$. Hence, $t \in G_C^f(r_C)$. Finally, notice that, if $r_C$ is a frugal repair of $I_C$, then there exists a
repair $r'$ of $I$ such that $G^f_C(r_C) = G^f(r')$. This concludes the proof. 
\end{proof}

\subsection{Proof of Separator Existence}
\label{sub:separator}

In this subsection, we prove \autoref{thm:separator_exists}. In particular, we show that 
the equivalence class $C^{sep} = \arg \min_{\text{sink } C \in \Eq} |coupled^{\oplus}(C)|$ is a 
separator. The proof has several  steps, and is the most technically involved part of this paper. 

The first step is to simplify our proof goal. Recall that we want to show that for any $C \in \Eq$,
where $C \neq C^{sep}$, either $C <^{\oplus} C^{sep}$ or $C \notin coupled^{\oplus}(C^{sep})$. We will show next that it suffices to consider only the sinks $C \in \Eq$, and show that for any sink 
$C \neq C^{sep}$, $C \notin coupled^{\oplus}(C^{sep})$. Indeed, we can show for a sink $C$, the set 
$coupled^{\oplus}(C)$ is {\em upward closed}: if $C_0 \in coupled^{\oplus}(C)$ and 
$C_0 <^{\oplus} C_1$, then also $C_1 \in coupled^{\oplus}(C)$. Note that $coupled^{\oplus}(C)$ 
is not necessarily upward closed for an arbitrary $C$.

\begin{lemma} 
\label{lem:non-sink:split}
If $C \in \Eq$ is a sink, then $coupled^{\oplus}(C)$ is upward closed. 
\end{lemma}

\begin{proof}
Assume that $C_0 \in coupled^{\oplus}(C)$ and $C_0 <^{\oplus} C_1$; we will show that
$C_1 \in coupled^{\oplus}(C)$.
Indeed, there exists a path $P: v_R \leftrightarrow u_S$ for $R \in C, S \in C_1$ such that $P \cap C^{\oplus} = \emptyset$. Since $C_0 <^{\oplus} C_1$,
there exists some $T \in C_1$ such that $u_T \in C_0^{\oplus}$. Thus, $u_T \in u_S^{\oplus}$ and
there exists a directed consistent path $P': u_S \rightarrow u_T$. Now, the path $P'' = P, P'$
connects $v_R$ with $u_T$. Notice that it is not possible that $P' \cap C^{\oplus} \neq \emptyset$, otherwise we would have that $u_T \in C^{\oplus}$, which contradicts the fact that $C$ is a sink. Hence, $P'' \cap C^{\oplus} = \emptyset$ and $C_1 \in coupled^{\oplus}(C)$.
\end{proof}

Now, suppose that we have shown that for any sink $C \neq C^{sep}$, $C \notin coupled^{\oplus}(C^{sep})$, and consider any $C' \in \Eq$, $C' \neq C$ that is not a sink. Then $C' <^{\oplus} C''$ for some $C'' \in \Eq$ that is a sink; hence, $C'' \notin coupled^{\oplus}(C^{sep})$. However, since $C^{sep}$ is a sink, we can apply \autoref{lem:non-sink:split} to conclude that $C' \notin coupled^{\oplus}(C^{sep})$. 

The bulk of the proof consists of two technical results. The first result tells us that for
a sink $C$, the two types of coupling coincide: $coupled^+(C) = coupled^{\oplus}(C)$.

\begin{proposition}
\label{prop:equal:strict:weak}
Let $G$ be a splittable and f-closed graph. For any sink $C \in \Eq$, $C^+ = C^{\oplus}$.
\end{proposition}

The second result tells us that for two distinct equivalence classes $C_1, C_2$  where
$C_1 \in coupled^+(C_2)$, $coupled^+(C_1)$ is strictly contained in $coupled^+(C_2)$.

\begin{proposition}
\label{prop:split:eqclass}
Let $G$ be a splittable graph and $C_1, C_2 \in \Eq$ such that $C_1 \neq C_2$. Then,
\begin{packed_enum} 
\item Either $C_1 \notin coupled^+(C_2)$ or $C_2 \notin coupled^+(C_1)$.
\item If $C_1 \in coupled^+(C_2)$, then $coupled^+(C_1) \subset coupled^+(C_2)$.
\end{packed_enum}
\end{proposition}

Now, consider a sink $C \neq C^{sep}$. If $C \in coupled^{+}(C^{sep})$, then by 
\autoref{prop:split:eqclass}(2) and \autoref{prop:equal:strict:weak} it must be that 
$coupled^{\oplus}(C^{sep}) = coupled^{+}(C^{sep}) \supset coupled^{+}(C) = coupled^{\oplus}(C)$, 
which contradicts the minimality of $coupled^{\oplus}(C^{sep})$, and this proves our
main theorem. In the rest of this section, we will present the proofs of \autoref{prop:equal:strict:weak} and \autoref{prop:split:eqclass}. 

We start with a proposition that will be used later.

\begin{proposition}
\label{prop:strict:cut}
If $C' \in  coupled^+(C)$ then there exists $R \in C$ such that for all $S \in C'$, $S \in coupled^+(R)$.
\end{proposition}

\begin{proof} 
For any node $v \in V(G)$, define
\begin{align} \label{eq:call}
 \mL_C(v) = \setof{T \in C}{v \in u_T^{+,T}}
  \end{align}
It is easy to see that $\emptyset \subseteq \mL_C(v) \subseteq C$. Moreover, 
$v \in C^+$ if and only if $\mL_C(v) = C$. 

Since $C' \in coupled^+(C)$, by definition there exists $R^0 \in C, S \in C'$ and a path 
$P^0 : v_{R^0} \leftrightarrow u_{S}$  such that $P^0 \cap C^+ = \emptyset$, or equivalently 
for every $v \in P^0$, $\mL_C(v) \subsetneq C$. We will show that there exists $R \in C$
and a path $P: v_{R} \leftrightarrow u_{S}$ such that $P \cap u_R^{+,R} = \emptyset$; this proves the proposition, since for any $S' \in C'$, there exists a directed path $P': u_{S'} \rightarrow u_S$ that cannot be cut by $u_{R}^{+,R}$ (otherwise it would be that $u_S \in u_R^{+,R}$, a contradiction to the fact that $P$ is not cut by $u_R^{+,R}$). 

If for every $v \in P^0$ we have that $v \notin u_{R^0}^{+,R^0}$, then our claim holds trivially for $R = R^0$ and $P = P^0$. Otherwise, there exists a node $v \in P^0$ such that $\mL_C(v) \supseteq \{R^0\} \supsetneq \emptyset$. If $P^0$ visits in order the nodes $v_{R^0} \equiv v_1, v_2, \dots, v_m \equiv u_S$, let $j$ be the largest index with the property that $\mL_C(v_j) \supsetneq \emptyset$. We thus have established that $\emptyset \subsetneq \mL_C(v_j) \subsetneq C$.

Since there exists an edge $T \in \mL_{C}(v_j)$, $v_j \in u_T^{+,T}$. Moreover, for any $U \in C$,
$u_T \in u_U^+$. Consequently, for any $U \in C$, $v_j \in u_U^+$.
But now, consider an edge $R \in C \setminus \mL_C(v_j)$ (such an edge always exists): since $v_j \in u_R^+ \setminus u_R^{+,R}$, $u_R$ reaches $v_j$ only by going through the edge $e_R$ first. Hence, there exists a simple path $P_j: v_R \leadsto v_j$ such that $P_j \cap u_R^{+,R} = \emptyset$.
Finally, let us construct the path $P: P_j, v_{j+1},( v_{j+1}, v_{j+2}),  \dots, (v_{m-1}, v_m ), v_m$ 
from $v_R$ to $u_S$.
By our construction, for all $i = j+1, \dots, m$ we have $\mL_{C}(v_i) = \emptyset$, and thus
$v_i \notin u_R^{+,R}$. Hence $P \cap u_R^{+,R} = \emptyset$.
\end{proof}

\begin{proof} (of \autoref{prop:split:eqclass})

\textbf{(1)}. Assume for the sake of contradiction that $C_1 \in coupled^+(C_2)$ and
$C_2 \in coupled^+(C_1)$. Then, from~\autoref{prop:strict:cut}, there exists some
$S_0 \in C_2$ such that for all $R \in C_1$, $R \in coupled^+(S_0)$ and some 
$R_0 \in C_1$ such that for all $S \in C_2$ we have $S \in coupled^+(R_0)$. In particular,
$R_0 \in coupled^+(S_0)$ and $S_0 \in coupled^+(R_0)$. But then, $R_0, S_0$ would
be an unsplittable pair that are not source-equivalent, a contradiction. 

\textbf{(2)}  From property (1), we obtain $C_2 \notin coupled^+(C_1)$. Since 
$C_2 \in coupled^+(C_2)$, it suffices to show  that $coupled^+(C_1) \subseteq coupled^+(C_2)$.
Indeed, let $C \in coupled^+(C_1)$, where $C \neq C_1, C_2$ (otherwise the claim is trivial).

Since $C \in coupled^+(C_1)$, by applying~\autoref{prop:strict:cut}, there exists edges $T \in C, R \in C_1$ and a path $P_{TR}: u_{T} \leftrightarrow v_R$ such that $P_{TR} \cap u_R^{+,R} = \emptyset$. 
Additionally, since $C_1 \in coupled^+(C_2)$, by applying~\autoref{prop:strict:cut}, we obtain that there exists $S \in C_2$ such that for every $R' \in C_1$, $R' \in coupled^+(S)$: in particular, $R \in coupled^+(R)$. Thus, there exists a path $P_{RS}: u_{R} \leftrightarrow v_S$ such that $P_{RS} \cap u_{S}^{+,S} = \emptyset$ (see \autoref{fig:proof:sketch}).

%
Construct now the path $P^+ = P_{TR}, e_R, P_{RS}$, which is of the form
$P^+: v_S \leftrightarrow u_T$. We will show that $P^+ \cap C_2^+ = \emptyset$, 
which proves that $C \in coupled^+(C_2)$. 

Suppose not; then, $P^+ \cap u_S^{+,S} \neq \emptyset$. Since the nodes of $P_{RS}$ do 
not intersect with $u_S^{+,S}$, there must exist a node $v \in P_{TR} \cap u_S^{+,S}$, which in 
turn implies the existence of a 
directed path  $P_S: u_S \leadsto v$ that does not contain the edge $e_S$. 
If $P_{Rv}$ denotes the fragment of the path $P_{RT}$ from node $v_R$ to node $v$, 
construct the path  $P^0 = P_{Rv}, P_S$ from $v_R$ to $u_S$. However, the fact that
$R \in C_1, S \in C_2$ and $C_2 \notin coupled^+(C_1)$ implies that  
$P^0 \cap C_1^+ \neq \emptyset$, and consequently $P^0 \cap u_R^{+,R} \neq \emptyset$.
But then, since $P_S$ is a directed path, 
$v \in u_R^{+,R}$, a contradiction to the fact that the path $P_{TR}$ does not intersect
with $u_R^{+,R}$. 
\end{proof}

\begin{figure}[tb]
  \centering
\resizebox{7cm}{!}{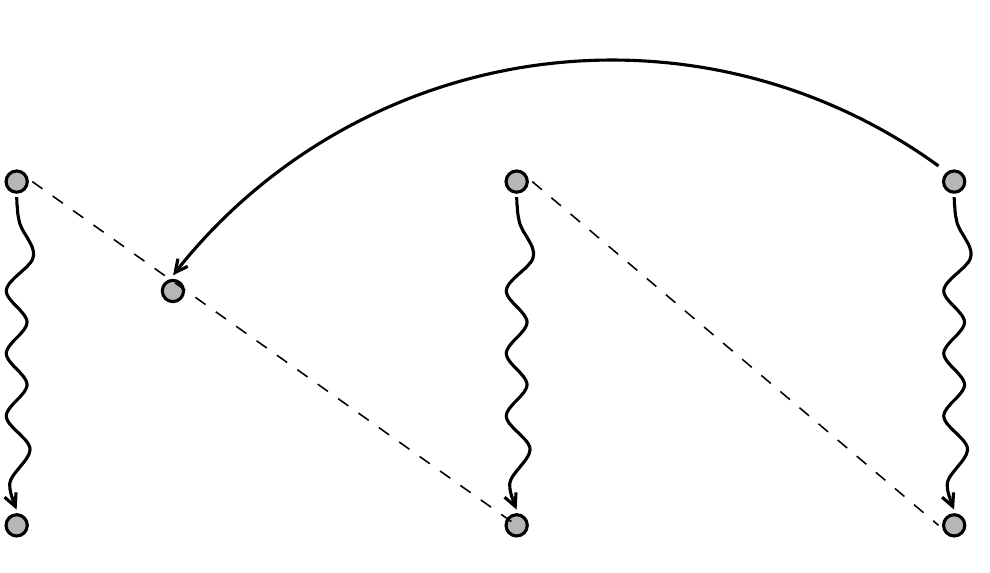}
\caption{{\footnotesize{The setting for the proof of~\autoref{prop:split:eqclass}, part (2).}}}
\label{fig:proof:sketch}
\end{figure}

\begin{proof} (of \autoref{prop:equal:strict:weak} )
Assume that $v \in C^+$; we will then show that $v \in C^{\oplus}$ as well. Let $A$ be the
unique SCC that contains the edges in $C$. 
We will first need the following lemma.

\begin{lemma}
\label{lem:help:1}
Let $v \in C^+$ and $v \in w^{\oplus}$ for some $w \in V(A)$. Then, $v \in C^{\oplus}$.
\end{lemma}

\begin{proof}
We will show that for any $u \in V(A)$, $v \in u^{\oplus}$ using induction on the distance
between $u,w$, denoted $d(u,w)$. For the basis of the induction, where distance is zero,
we have that $v \in w^{\oplus}$ and $d(w,w)=0$.

For the induction step, consider some node $u$ with $d(u,w) = d+1$. Then, there exists an
edge $e_T$ such that $u = u_T$ and $d(v_T,w) = d$. By the induction hypothesis,
$v \in v_T^{\oplus}$. If $T$ is consistent relation, then trivially $v \in u_T^{\oplus}$. Otherwise,
$T \in C$, and since $v \in C^+$, $v \in u_T^{+,T}$. Since $G$ is f-closed, this implies
that $v \in u_T^{\oplus}$.
\end{proof}

We now distinguish two cases for some $v \in C^+$. If $v \in V(A)$, then 
$v \in v^{\oplus}$ and thus by \autoref{lem:help:1},
$v \in C^{\oplus}$. Otherwise, $v \in V(G) \setminus V(A)$.
Since $v \in C^+$, there exists a directed path $P: w \leadsto v$ such that $P \cap V(A) = \{w\}$. Let 
$P$ visit in sequence the nodes $w = v_0, v_1, \dots, v_m = v$ and notice that if $i \leq j$, $\mL_C(v_i) \subseteq \mL_C(v_j)$ ($\mL_C$ as defined in~\eqref{eq:call}).
Since $\mL_C(v_m) = C$, let $k$ be the first index such that $\mL_C(v_k) = C$. We will next show
that $v_k \in w^{\oplus}$, which implies that $v_k \in C^{\oplus}$. Since $C$ is a sink all the edges
$(v_{i}, v_{i+1})$ for $i=k, \dots, m-1$ must be consistent and thus $v = v_m \in C^{\oplus}$ as well. 

By the choice of $v_k$, there exists some $S \in \mL_C(v_k)$ such that for any $i < k$, 
$S \notin \mL_C(v_i)$. Additionally, since $S \in \mL_C(v_k)$, there exists a path 
$P': u_S \leadsto v_k$. Let $w'$ be the last node 
of the path $P'$ inside $A$; we know that $w \neq w'$, since $S \notin \mL_C(w)$. Finally,
let $P'_{w'}$ be the part of the path $P'$ from $w'$ to the first node $v_k$. 
The important observation is no node of $P'_{w'}$ will be in the same SCC as nodes 
$v_1, \dots, v_{k-1}$, since otherwise $S \in \mL_{C}(v_i)$ for some $i <k$. 
So, now we can create the following 2 paths from $w$ to $v_k$:
the first path $P_1$ follows $P$ from $w$ to $v_k$, while the second path $P_2$
follows the simple path inside $A$ from $w$ to $w'$ and then $P'_{w'}$.
By our previous argument, for any $w \in P_2 \setminus \{w, v_k\}$, $w$ does not belong in
the same SCC with any of the nodes in $P_1 \setminus \{w, v_k\}$. We can now apply
the following lemma to conclude that $v_k \in w^{\oplus}$. 

\begin{lemma}
\label{lem:double:paths}
Let $G$ be a f-closed and splittable graph. If for $u,v \in V(G)$ there exist two directed simple 
paths $P_A, P_B : u \leadsto v$ such that no $w_A \in P_A \setminus \{u,v\}$, 
$w_B \in P_B \setminus \{u,v\}$ are in the same SCC, then $v \in u^{\oplus}$. 
\end{lemma}

\begin{proof}
Let $P_A$ visit in order the nodes $u = w_A^1, \dots, w_A^m = v$ and 
similarly $P_B$ the nodes $u = w_B^1, \dots, w_B^{m'} = v$.
We will show that for any pair $w_A^i, w_A^j$, $v \in w_A^{i \oplus} \cup w_B^{j \oplus}$.
This proves the lemma, because we can choose $i = j = 1$. 
Suppose the claim does not hold, and consider the pair $w_A^i, w_B^j$ such that
$v \notin w_A^{i \oplus}, v \notin w_B^{j \oplus}$ and $i+j$ is maximum.
First, note none of these nodes is $v$, otherwise $v \in w_A^{i \oplus} \cup w_B^{j \oplus}$ trivially. 
Next, consider any node $w_A^k$, for $k>i$. Then, it must be that $v \in w_A^{k \oplus}$, since
the pair $w_A^k, w_B^j$ has $k +j > i+j$ and $v \notin w_B^{j \oplus}$. Similarly for any node
$w_B^k$ with $k > j$, $v \in w_B^{k\oplus}$.

Hence, both edges $e_R = (w_A^i, w_A^{i+1})$ and $e_S = (w_B^j, w_B^{j+1})$ must be inconsistent.
So, $R,S \in E^i$ and $v \in v_R^{\oplus}, v \in v_S^{\oplus}$. Moreover, if $u_R = u$
then $v \in u_{R}^{+,R}$ as well, which implies that $v \in u_R^{\oplus} = w_A^{i\oplus}$,a contradiction
(similarly, $u_S \neq u$). By our assumption of the path structure, $u_R, u_S$ do not belong in
the same SCC and thus $R \nsim S$. 

Finally, let $P_{RS}$ be the path that visits in order the nodes 
$v_R = w_A^{i+1}, \dots, w_A^m = v = w_B^{m'}, \dots, w_B^j = u_S$ and symmetrically $P_{SR}$
the path that visits $v_S ,\dots, u_R$. Since $G$ is splittable and $R \nsim S$, either 
$P_{RS} \cap u_{R}^{+,R} \neq \emptyset$ or $P_{SR} \cap u_{S}^{+,S} \neq \emptyset$. W.l.o.g., assume
that $w \in P_{RS} \cap u_{R}^{+,R}$. Since every node in $P_{RS}$ has directed path to $v$ that
does not go through $e_R$, this further implies that $v \in u_R^{+,R}$. Additionally, we have
already shown that $v \in v_R^{\oplus}$. Since $G$ is $f$-closed, this immediately implies that
$v \in u_R^{\oplus}$, which is a contradiction to the existence of the pair.
\end{proof}

This concludes the proof.
\end{proof}

\section{The coNP-complete Case}
\label{sub:coNP_case}

In this section, we prove part (2) of \autoref{thm:dichotomy}: if
$G[Q]$ is unsplittable, then \cqa{$Q$} is coNP-complete.  We reduce
\cqa{$Q$} from \textsc{Monotone-3Sat}, which is a special case of
\textsc{3Sat} where each clause contains only positive or only
negative literals. We say that a clause is positive (negative) if it
contains only positive (negative) literals.
\textsc{Monotone-3Sat} is known to be a NP-complete problem \cite{GJ79}. 

Given an instance $M$ of \textsc{Monotone-3Sat}, let us denote by $\Phi$
the set of all clauses, $X$ the set of all variables, $X^*$ the set of
all literals and $\bB = \{T,F\}$ (true, false).  Moreover, let us
define $\top = \Phi \times \bB = \setof{(\phi, x^*)}{x^* \in \phi,
  \phi \in \Phi}$ and $\bot = \set{()}$.  We order the set $\mL =
\set{\bot, \bB, X, \Phi, X^*, \top}$ as shown in
\autoref{fig:label_lattice}: $\bot$ and $\top$ are the minimal and
maximal elements, and $\bB \leq \Phi$, $X \leq X^*$ and $\bB \leq
X^*$.  The reader may check that $\mL$ is a lattice.  For example,
$\Phi \wedge X^* = \bB$ and $\bB \vee X = X^*$.


\begin{figure}
\centering
\resizebox{5cm}{!} {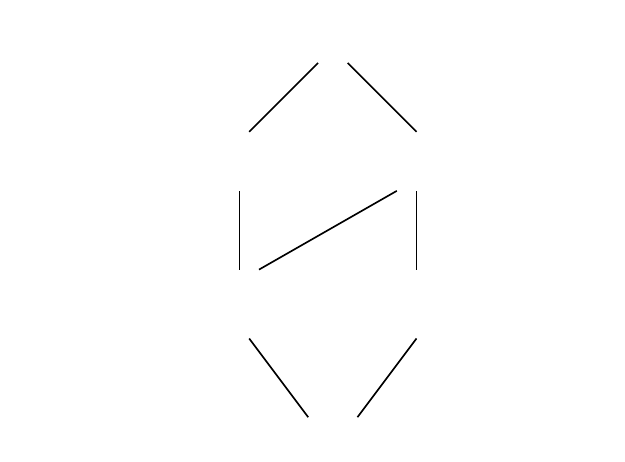} 
\caption{The lattice of the set of labels $\mL$.}
\label{fig:label_lattice}
\end{figure}


\begin{definition}[Valid Labeling]
\label{def:valid:labeling}
Let $R,S \in E^i$. A labeling $L: V(G) \rightarrow \mL$ is {\em $(R,S)$-valid} 
if the following conditions hold:
\begin{packed_enum}
\item \label{item:valid:1} $L(u_R) = \Phi$ and $L(v_R) \in \{\top, X, X^*\}$.
\item \label{item:valid:2} $L(u_S) = X$ and $L(v_S) \in \{\bB, X^*\}$.
\item \label{item:valid:3} For every $T \in E^i \setminus \{R,S\}$, $L(u_T) \geq L(v_T)$.
\item \label{item:valid:4}  $\exists P_R: v_R \leftrightarrow u_S$ such that $\forall v \in P_R, L(v) \geq X$.
\item \label{item:valid:5} $\exists P_S: v_S \leftrightarrow u_R$ such that $\forall v \in P_S, L(v) \geq \bB$.
\end{packed_enum}
\end{definition}

We prove:

\begin{proposition} 
\label{prop:valid:labeling}
If $R,S \in E^i$ are coupled and $S \not \lesssim R$, then $G$ admits a $(R,S)$-valid labeling.
\end{proposition}
 
If the query $Q$ has an unsplittable graph $G = G[Q]$, then there
exists two coupled edges $R, S$ s.t.  $R \nsim S$.  This implies that
we cannot have both $R \lesssim S$ and $S \lesssim R$, and the
proposition tells us that $G$ has an $(R,S)$-valid labeling.  We will
show later how to use this labeling to reduce $M$ to \cqa{$Q$}.
First, we prove the proposition.


\begin{proof} Since $S \in coupled^+(R)$, there exists a path $P_R :
  v_R \leftrightarrow u_S$ s.t. $P_R \cap u_R^{+,R} = \emptyset$;
  similarly, there exists a path $P_S : v_S \leftrightarrow u_R$
  s.t. $P_S \cap u_S^{+,S} = \emptyset$.  Notice that, in particular,
  $P_R$ contains the source and destination nodes $v_R, u_S$, and,
  similarly, $P_S$ contains the nodes $v_S, u_R$, which implies:
  \begin{align}
    v_R \not\in u_R^{+,R} && u_S \not \in u_R^{+,R} && v_S \not \in u_S^{+,S} && u_R \not \in u_S^{+,S} \label{eq:notin}
  \end{align}

  We define the label $L$ as follows.  Let $W = \set{u_R,v_R,u_S,v_S}$
  and set the initial labels for the four nodes in $W$:
  \begin{align*}
    L_0(u_R) = & \Phi, & L_0(v_R) = & \top, & L_0(u_S) = & X, & L_0(v_S) = & X^*
  \end{align*}
  For every node $v \in V(G)$, denote $\mW^{-1}(v) = \setof{w}{w \in
    W, v \in w^{+,R,S}}$, where $w^{+,R,S}$ is the set of nodes
  reachable from $w$ by a directed path that does not go through
  either $R$ or $S$.  In other words, $\mW^{-1}(v)$ is the subset of
  the four distinguished nodes that can reach $v$ without using $R$ or
  $S$.  Trivially, $w \in \mW^{-1}(w)$, for every $w \in W$.  Define
  the labeling $L$ as follows:
\begin{align*}
 \forall v \in V(G):  \quad L(v) = \bigwedge \setof{L(w)}{w \in \mW^{-1}(v)}
\end{align*}
We will show that this labeling is $(R,S)$-valid.  We start by
checking properties (\ref{item:valid:1}) and (\ref{item:valid:2}).
Consider each of the four distinguished nodes in $W$:
\begin{packed_grep}
\item[$u_R$:] The set $\mW^{-1}(u_R)$ is either $\set{u_R}$ or
  $\set{u_R, v_R}$; indeed $v_S \not\in \mW^{-1}(u_R)$ because $S \not
  \lesssim R$, and $u_S \not\in \mW^{-1}(u_R)$ by Eq.(\ref{eq:notin}).
  By definition, either $L(u_R)=\Phi$ or $L(u_R) = \Phi \wedge \top =
  \Phi$; in both cases $L(u_R)=\Phi$.
\item[$u_S$:] We have $\set{u_S} \subseteq \mW^{-1}(u_S) \subseteq
  \set{u_S, v_R, v_S}$, because Eq.(\ref{eq:notin}) implies $u_S
  \notin u_R^{+,R,S}$.  This implies $X = L_0(u_S) \geq L(u_S) \geq
  L_0(u_S) \wedge L_0(v_R) \wedge L_0(v_S) = X \wedge \top \wedge
  X^*=X$, hence $L(u_S) = X$.
\item[$v_R$:] We have $\set{v_R} \subseteq \mW^{-1}(v_R) \subseteq
  \set{u_S, v_R, v_S}$, because Eq.(\ref{eq:notin}) implies $v_R
  \notin u_R^{+,R,S}$.  Therefore, $\top \geq L(v_R) \geq X \wedge
  \top \wedge X^* = X$, implying $L(v_R) \in \set{X, X^*, \top}$.
\item[$v_S$:] We have $\set{v_S} \subseteq \mW^{-1}(v_S) \subseteq
  \set{u_R, v_R, v_S}$, because Eq.(\ref{eq:notin}) implies $v_S
  \notin u_S^{+,R,S}$.  Therefore, $X^* \geq L(v_S) \geq \Phi \wedge
  \top \wedge X^* = \bB$, implying $L(u_S) \in \{\bB, X^*\}$.
\end{packed_grep}

To show property (\ref{item:valid:3}), consider an edge $e_T =
(u_T,v_T)$, $T \neq R,S$.  Then $\mW^{-1}(u_T) \subseteq
\mW^{-1}(v_T)$ which implies $L(u_T) \geq L(v_T)$.

For (\ref{item:valid:4}), let $P_R$ be the undirected path defined
earlier s.t.  $P_R \cap u_R^{+,R} = \emptyset$; we also have $P_R \cap
u_R^{+,R,S}= \emptyset$.  Let $v \in P_R$ be any node on this path.
Then $u_R \notin \mW^{-1}(v)$, which implies that $\mW^{-1}(v)
\subseteq \set{v_R, u_S, v_S}$, and therefore $L(v) \geq \top \wedge X
\wedge X^* = X$.

Finally, for  (\ref{item:valid:4}), let $P_S$ be the undirected path
defined earlier, s.t. $P_S \cap u_S^{+,S} = \emptyset$.  As before,
for any node $v \in P_S$ we have $\mW^{-1}(v) \subseteq \set{u_R, v_R,
  v_S}$, and therefore $L(v) \geq \Phi \wedge \top \wedge X^* = \bB$.
\end{proof}

Next, we show how to use a valid labeling to reduce the
\textsc{Monotone-3Sat} $\Phi$ to \cqa{$Q$}.

 
\introparagraph{The Functions $f_{L_1L_2}$} For any pair of sets $L_1,
L_2 \in \mL$ such that $L_1 \geq L_2$, we define a function
$f_{L_1L_2}: L_1 \rightarrow L_2$, as follows.  First, for the seven
pairs $L_1, L_2$ where $L_1$ covers\footnote{In a lattice, $L_1$ {\em
    covers} $L_2$ if $L_1 > L_2$ and there is no $L_3$ s.t. $L_1 > L_3
  > L_2$.} $L_2$, we define $f_{L_1L_2}$ directly:
\begin{packed_grep}
\item[$(\Phi, \bB)$]: $f_{\Phi, \bB}(\phi) = T$ if $\phi$ is a positive clause, else $F$
\item[$(X^*, X)$]: $f_{X^*,X}(x^+) = f_{X^*,X}(x^{-}) = x$
\item[$(X^*, \bB)$]: $f_{X^*,\bB}(x^+) = T$ and $f_{X^*,\bB}(x^-) = F$
\item[$(\top, \Phi)$]: $f_{\top, \Phi}((\phi,x^*)) = \phi$ 
\item[$(\top, X^*)$]: $f_{\top, X^*}((\phi,x^*)) = x^*$
\item[$(\bB,\bot),(X,\bot)$]: $f_{\bB,\bot}(b) = f_{X,\bot}(x) = ()$
\end{packed_grep}
Next, we define $f_{LL} = id_L$ (the identity on $L$) and $f_{L_1L_3}
= f_{L_2L_3} \circ f_{L_1L_2}$ for all $L_1 \geq L_2 \geq L_3$.
Readers familiar with category theory will notice that we have
transformed the lattice $\mL$ into a category.

\introparagraph{Instance Construction} Now we define the database
instance $I$, by defining a binary relation $T^I$ for every relation
name $T$.  Let $L_1 = L(u_T)$, $L_2 = L(v_T)$ be the labels of the
source and target node of $e_T$.
We distinguish two cases, depending on whether $T$ is $R,S$ or not. 

If $T \neq R$, $T \neq S$, then we know that $L_1 \geq L_2$.  Define
$T^I = \setof{(a,b)}{a \in L_1, b = f_{L_1L_2}(a) \in L_2}$.  Notice
that the first attribute of $T^I$ is a key (because $f_{L_1L_2}$ is a
function), and therefore we ensure that $T^I$ always satisfies the key
constraint, regardless of whether the type of $T$ was consistent or
inconsistent.

If $T=R$ or $T=S$, then $L_1 \not\geq L_2$.  In this case we construct $R^I$ and
$S^I$ to be a certain set of pairs $(a,b)$, $a \in L_1, b \in L_2$,
where $b$ is obtained from $a$ by either going ``back'' in the
lattice, or going ``back and forth'', depending on the particular
combination of $L_1, L_2$ given by \autoref{def:valid:labeling}:

\begin{packed_grep}
\item[$(\Phi,\top$)]: $R^I = \setof{(a,b)}{b \in f^{-1}_{\top, \Phi}(a)}$ (back)
\item[$(\Phi,X^*$)]: $R^I = \setof{(a,b)}{\exists c \in f^{-1}_{\top,\Phi}(a): f_{\top,X^*}(c) = b}$ (back-and-forth)
\item[$(\Phi,X$)]: $R^I = \setof{(a,b)}{\exists c \in f^{-1}_{\top,\Phi}(a):f_{\top,X}(c) = b}$ (back)
\item[$(X,X^*$)]: $S^I = \setof{(a,b)}{b \in f^{-1}_{X^*,X}(a)}$ (back)
\item[$(X,\bB$)]: $S^I = \setof{(a,b)}{\exists c \in f^{-1}_{X^*,X}(a): f_{X^*,\bB}(c) = b}$ (back-and-forth)
\end{packed_grep}

Notice that in all cases $R^I$ and $S^I$ are inconsistent.  In the
first case, a repair of $R^I$ chooses for each clause $\phi \in \Phi$
a value $(\phi, b)$ with $b \in \bB$; in the second case, a repair of
$R^I$ chooses for each clause $\phi$, a literal $x^* \in \phi$, while
in the third case a repair chooses for each clause $\phi$ a variable
$x$ in that clause.


\begin{example}
  Consider the formula $Y = \phi_1 \wedge \phi_2$, where $\phi_1 =
  (x^+ \vee y^+ \vee z^+)$ and $\phi_2 =(z^- \vee w^- \vee t^-)$. If
   the inconsistent relation $R$ is labeled with $(\Phi, X)$, it will be
  populated by the tuples $(\phi_1,x), (\phi_1,y), (\phi_1,z)$ and
  $(\phi_2,z), (\phi_2,w), (\phi_2,t)$. On the other hand, a 
  consistent relation $T \neq R,S$ that is labeled with $(\Phi, \bB)$
  will contain the tuples $(\phi_1, T), (\phi_2, F)$.
\end{example}

Thus, given a valid labeling we can create a database instance using the construction 
we just presented.  We prove:

\begin{proposition}
\label{prop:3sat_reduction}
Let $I$ be the instance that corresponds to a $(R,S)$-valid labeling according to an instance $M$ of 
\textsc{Monotone-3Sat}. Then, $I \not \vDash Q$ if and only if $M$ has a satisfying assignment. 
\end{proposition} 

\begin{proof}
First, note that the valid labeling guarantees that, if $T \neq R,S$, then $T$ will be a consistent relation
in the instance $I$. On other other hand, the relations $R$ and $S$ will be inconsistent.

Consider a satisfying assignment for $M$, where $v(x)$ denotes the assigned value (true or false) for variable $x$. We will construct a repair $r$ that does not satisfy $Q$. Since the assignment satisfies the formula, for every clause $\phi$ there exists a literal $x^*$ that evaluates to true. Then, for the relation $R$, $r$ includes the tuple $(\phi, x^*)$ (if $e_R$ has labels $(\Phi,X^*)$) or $(\phi,x)$ (if $(\Phi ,X)$) or $(\phi, (\phi, x^*))$ (if $(\Phi, \top)$). As for the relation $S$, we have two cases. 
If the labels are $(X,A)$, $r$ includes $(x,F)$ when $v(x)=T$, and $(x,T)$ when $v(x)=F$. 
Similarly, for $(X ,X^*)$, if $v(x) = T$, $r$ includes the tuple $(x, x^-)$, otherwise
if $v(x)=F$, $r$ includes $(x,x^+)$.

It remains to show that $Q(r)$ evaluates to false. For the sake of contradiction, assume that $Q(r)$ is true  and consider a tuple $t \in Q^f(r)$. Let $t[u_R]=\phi$ and assume w.l.o.g. that it is a positive clause. Then,
$t[v_R] \in \{x, x^+, (\phi, x^+)\}$, for a variable $x$ with assignment $v(x)=T$.
Note that there must be a path from $v_R$ to $u_S$ such that every label 
in the path has a consistent mapping to $X$. Hence, $t[u_S] = x$, which implies that $t[v_S] \in \set{F,x^-}$ by our construction of $I$.
But this is a contradiction, since there exists a path from $v_S$ ($t[v_S] \in \set{F,x^-}$) to $u_R$ 
($t[u_R] = \phi$ is a positive clause), where each label has a consistent mapping to $A = \{T,F\}$.  

For the inverse direction, assume that $I$ has a repair such that $Q(r)$ is false. 
We construct an assignment for the variables in $M$ as follows: if the repair $r$ contains a tuple 
$(x, T)$ (or $(x,x^+)$) in relation $S$, we let $v(x) = F$; otherwise, $v(x) = T$. 
Now, consider a positive (w.l.o.g.) clause $\phi$ of the instance $M$. Assume that $r$ contains in $R$ a
tuple $(\phi,x)$ (or $(\phi, x^+)$ or $(\phi,(\phi,x^+))$). 
Using similar arguments as before, one can see that $r$ cannot include $(x,T)$ (or $(x,x^+)$); 
otherwise, $Q(r)$ would evaluate to true. Hence, $v(x) = T$ and clause $\phi$ will be satisfied.
\end{proof}

\begin{figure}
\centering
\resizebox{5cm}{!} {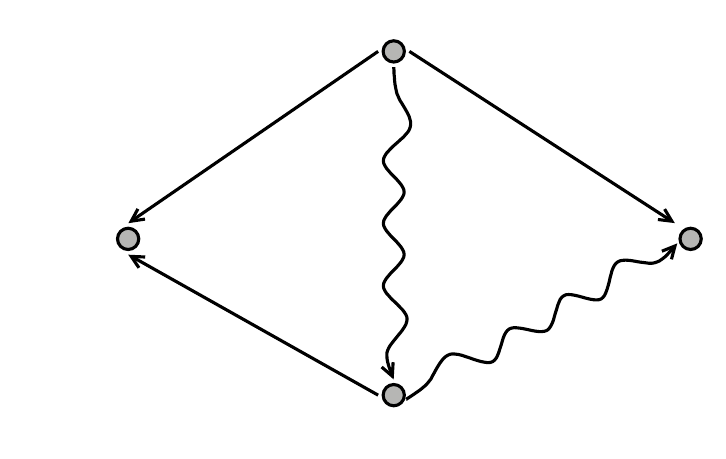} 
\caption{A query graph with a $(R,S)$-valid labeling.}
\label{fig:label_example}
\end{figure}

\begin{example}
As an example of the labeling construction, consider the query of 
\autoref{fig:label_example}.
Notice that $R \lesssim S$. Also, $u_R = x$, $v_R = u_S = y$ and $v_S=z$.
Since $L^+(x) = \{L_0(u_R)\} = \{\Phi\}$, $L(x)  = \Phi$. Also,
$L^+(y) = \{L_0(v_R), L_0(u_S)\} = \{\top,X\}$, hence
$L(y) = \top \wedge X = X$. For variable $z$, $L^+(z) = \{L_0(v_S), L_0(u_R)\} = \{\Phi,X^*\}$
and $L(z) = \Phi \wedge X^* = \bB$. 
Finally, $L^+(t) = \{L_0(u_R), L_0(v_R), L_0(u_S)\} = \{\Phi, \top, X\}$ and hence 
$L(t) = \Phi \wedge \top \wedge X = \bot$.
\end{example}

\section{Related Work}
\label{sec:related}

The consistent query answering framework was first proposed by Arenas et al. in~\cite{ABC99}. Fuxman and Miller~\cite{FM07} focused on primary key constraints, with the goal of specifying conjunctive queries where \cqa{$Q$} is first-order expressible, \ie can be represented as a boolean first-order query over the inconsistent database. They presented a class of acyclic conjunctive queries w/o self-joins, called $C_{forest}$, that allows such first-order rewriting. Further,  Fuxman et al.~\cite{FFM05} designed and built a system that supported the query rewriting functionality for consistent query answering.

In a series of papers~\cite{Wijsen09, Wijsen10}, Wijsen improved on the results for first-order expressibility. The author presented a necessary and sufficient syntactic condition for the first-order expressibility for acyclic conjunctive queries without self-joins. In a later paper, Wijsen~\cite{Wijsen10b} gave a polynomial time algorithm for the query $Q_2 = R(\ux, y), S(\uy,x)$, which is known to be not first-order expressible. $Q_2$ is the first query that was proven to be tractable even though it does not admit a first-order rewriting. Kolaitis and Pema~\cite{KP12} proved a dichotomy for the complexity of \cqa{$Q$} when the query has only two atoms and no self-joins into polynomial time and coNP-complete.
Finally, Wijsen~\cite{Wijsen13} recently classified several acyclic queries into PTIME
and coNP-complete, without however showing the complete dichotomy for acyclic queries without
self-joins.

A relevant problem to consistent query answering is the counting version of the problem: given a query and an inconsistent database, count the number of repairs that satisfy the query. Maslowski and Wijsen~\cite{MW11} showed that this problem admits a dichotomy in P and \#P-complete for conjunctive queries without self-joins.

Finally, we should mention that the problem of consistent query answering is closely related to probabilistic databases, and in particular {\em disjoint-independent} probabilistic databases~\cite{DS07}. Wijsen in~\cite{Wijsen13} discusses the precise connection between the complexity of evaluating a query $Q$ on a probabilistic database and \cqa{$Q$}. 

\section{Conclusion}
\label{sec:conclusion}

In this paper, we make significant progress towards proving a dichotomy on the complexity of \cqa{$Q$},  studying the case where $Q$ is a Conjunctive Query without self-joins consisting of atoms with simple keys or  keys containing all attributes.  It remains a fascinating open question whether such a dichotomy exists for general conjunctive queries, even in the simpler case where there are no self-joins. 

\bibliographystyle{abbrv}
\bibliography{ref}  

\newpage
\appendix

\section{Simplifying the Structure}
\label{sec:simplify:structure}

In this section, we show how to transform any query that consists of atoms where the key is either a single attribute or all attributes to a query which we call {\em graph-representable}.

\begin{definition}
A boolean connected CQ $Q$ is {\em graph-representable} 
if it is w/o self-joins, w/o constants, w/o duplicate variables in a single atom, and further
contains only binary atoms where each atom has exactly one attribute as key.
\end{definition}

First, note that we can assume w.l.o.g. that the hypergraph for $Q$ is connected; 
otherwise, we can solve \cqa{$Q$} for each of the connected components and decide that $Q$ is certain if and only if every component is certain.

We write that $\cqa{Q} \stackrel{FO}{\sim} \cqa{Q'}$ if there exists an FO-expressible reduction from
\cqa{$Q$} to \cqa{$Q'$} and vice versa.

\begin{theorem}
\label{thm:reduction}
Let $Q$ be a connected boolean CQ without self-joins, where the key for each atom is either a single attribute or all attributes. Then, there exists a graph-representable query $Q'$ such that
 $\cqa{Q} \stackrel{FO}{\sim} \cqa{Q'}$.
\end{theorem}

The FO-expressible reduction described in the above theorem can be decomposed in a sequence
of simpler steps, which we describe next, thus proving \autoref{thm:reduction}.

In the case where a query $Q$ contains an atom $R$ with constants and/or variables that
appear twice, we can reduce the query $Q$ to a query $Q'$ where $R$ is replaced by an atom $R'$
that contains only variables that appear exactly once.

\begin{proposition}
Let $Q$ be a CQ  that contains an atom $R$. Let $Q'$ be the query where we have replaced $R$
with an atom $R'$ without constants, and where every variable appears exactly once. Then, 
$\cqa{Q} \stackrel{FO}{\sim} \cqa{Q'}$.
\end{proposition}

We can further simplify the query structure by removing unary relations.

\begin{proposition}
\label{prop:unary}
Let $Q$ be a connected CQ and $Q'$ be the query derived from 
$Q$ by removing all occurrences of unary atoms.
Then, $\cqa{Q} \stackrel{FO}{\sim} \cqa{Q'}$.
\end{proposition}

\begin{proof}
Notice that every unary relation is consistent by definition, since the only attribute is the primary key.  
Let $U(x)$ be such a unary relation in $Q$ and consider another appearances of variable $x$ in the query. 
Consider any atom that contains $x$ as a variable. Then, by \autoref{lem:frugal}, we can remove from this atom any key-group such that $x$ assumes a value $a$ and $a \notin U^D$, since no frugal repair will contain $a$ in the answer set. After this processing of $I$, $U$ plays no role in whether a repair satisfies the query and hence can be removed to obtain a query $Q_{-U}$ without the atom $U(x)$. Notice also that the processing is FO-expressible. For the inverse reduction, we can add a unary relation $U(x)$ to $Q_{-U}$ such that $U = \Pi_{x} (Q_{-U}^f(I))$ (since $Q$ is connected, $Q_{-U}$ always contains an appearance of variable $x$). 
\end{proof}

Next, we show how to handle the atoms where the primary key consists of all the attributes: such an example could be $R(\underline{x,y})$ or $S(\underline{x,y,z})$. In the general setting, we are given an atom of the form $R(\underline{x_1, \dots, x_k})$. Observe that the relation $R$ will be always consistent, since it is not possible to have any key violations. 

\begin{proposition}
\label{prop:allkey}
Let $Q$ be a connected CQ containing $R(\underline{x_1, \dots, x_k})$,
and let $Q'$ be the query obtained by replacing $R$ in $Q$ with $k$ new consistent relations
 $R_1^c(\ux, x_1), \dots, R_k^c(\ux, x_k)$ ($x$ is a new variable that does not appear in $Q$).
Then, $\cqa{Q} \stackrel{FO}{\sim} \cqa{Q'}$.
\end{proposition}

\begin{proof}
To reduce \cqa{$Q'$} to \cqa{$Q$}, we simply compute $R(\underline{x_1, \dots, x_k})$ as the natural join of the relations $R_1, \dots, R_k$ on the common variable $x$, where we have projected out the joining variable $x$. For the inverse reduction, we populate $R_1, \dots, R_k$ by introducing, for every tuple $R(\underline{a_1, \dots, a_k})$, $k$ new tuples $R_1(\underline{(a_1, \dots, a_k)}, a_1), \dots, R_k(\underline{(a_1, \dots, a_k)}, a_k)$. It is easy to see that every $R_i$ is a consistent relations where the variable $x$ is the primary key. Additionally, the two instances are equivalent w.r.t the repairs they admit.
\end{proof}

It now remains to deal with the case of relations that have arity $\geq 3$ and additionally have a single variable as primary key. For this, we need the following lemma.

\begin{proposition}
\label{prop:arity3}
Let $Q$ be a connected CQ containing $R(\ux, y_1, \dots, y_k)$. Denote by $Q^s$ the query obtained by replacing $R$ with $R'(\ux,y), S_1^c(\uy, y_1), \dots, S_k^c(\uy, y_k)$, where $y$ is a new variable. Then, $\cqa{Q} \stackrel{FO}{\sim} \cqa{Q'}$.
\end{proposition}

\begin{proof}
For the one direction, assume we have query $Q$, along with a database instance $I$. We transform $I$ to an instance $I^s$ for query $Q^s$ as follows. For a tuple $R(\underline{a}, b_1, \dots, b_k)$, we introduce in $I^s$ the tuple $R'(\underline{a}, (b_1, \dots, b_k))$ and also, for $i=1, \dots, k$ the tuples $S_i(\underline{(b_1, \dots, b_k)}, b_i)$. Observe that our construction guarantees that $S_i$ are consistent relations. It suffices to show that $I \vDash Q$ iff $I^s \vDash Q^s$. Notice that there is a one-to-one correspondence between repairs of $I, I^s$. Indeed, if some repair $r$ of $I$ chooses $R(\underline{a}, b_1, \dots, b_k)$, the corresponding repair $r^s$ of $I^s$ will choose $R(\underline{a}, (b_1, \dots, b_k))$ and vice versa. Now, observe that if $Q(r)$ evaluates to true, so will $Q(r^s)$ and vice versa.

For the inverse direction, assume $Q^s$ and an instance $I^s$. We transform $I^s$ to an instance $I$ of $Q$ by constructing $R(x, y_1, \dots, y_k) = R'(x,y), S_1(y, y_1), \dots, S_k(y, y_k)$, \ie in order to construct $R$, we join all relations on $y$ and then project out $y$. We will show that $I \vDash Q$ iff $I^s \vDash Q^s$. First, assume  that $I^s \vDash Q^s$; we will show that $I \vDash Q$. Indeed, consider a repair $r$ of $I$ and construct a repair $r^s$ that makes the same choices as $r$ for all common relations between $Q,Q^s$ and, if $R(\underline{a}, b_1, \dots, b_k) \in r$, then $R'(\underline{a},b) \in r^s$ for some $b$ such that $S_i(\underline{b}, b_i) \in I^s$ for every $i=1, \dots, k$ (note that by our construction such a $b$ always exists). Since $Q(r^s)$ is true, $Q(r)$ will be true as well. 

For the inverse, assume $I \vDash Q$ and consider a repair $r^s$ of $I^s$. Notice first that, for a key group $R'(\underline{a}, -)$, if $R'(\underline{a},b)$ and $\exists i: S_i(\underline{b},-) \notin I^s$, $a$ will never contribute towards an answer for $Q^s$, hence we can throw away w.l.o.g. such a key-group from consideration. 
Let $R'(\underline{a},-)$ be any key group in $I^s$; equivalently, $R(\underline{a},-)$ is a key group in $I$. Now, let $R'(\underline{a},b)$ be the unique tuple in $r^s$ from this key-group. As we have argued, there exist tuples $S_1(\underline{b},b_1), \dots, S_k(\underline{b},b_k)$ in $I^s$ (and $r^s$) and these tuples are unique.  By our construction, the instance $I$ contains the tuple $R(\underline{a},b_1, \dots ,b_k)$: this is the tuple that we include in $r$. Since $Q(r)$ evaluates to true, so must $Q^s(r^s)$. 
\end{proof}

The combination of the above propositions proves \autoref{thm:reduction}.

\begin{lemma}
\label{lem:frugal}
Let $a$ be value that does not appear in $\mathcal{M}_Q(I)$, and let $I^{-a} \subseteq I$  
s.t. every key-group that contains $a$ has been removed. Then, $I \vDash Q$ iff $I^{-a} \vDash Q$.
\end{lemma}

\end{document}